\documentclass[journal]{IEEEtran}
\ifCLASSINFOpdf
\else
   \usepackage[dvips]{graphicx}
\fi
\usepackage{url}
\usepackage{graphicx}
\usepackage{stfloats}
\usepackage{fancyhdr}
\usepackage{multirow, cite}
\usepackage{xcolor,colortbl, soul}
\usepackage{tabularx,booktabs}
\usepackage[utf8]{inputenc} 
\usepackage[T1]{fontenc}
\usepackage{url}
\usepackage{ifthen}
\usepackage{cite}
\usepackage[cmex10]{amsmath} 
\usepackage{amsmath,amssymb,amsfonts,amsthm,bm}
\usepackage{subfigure}
\usepackage{cite}
\usepackage{amsmath,amssymb,graphicx,epsfig,fancyhdr,amsthm,tabulary,epstopdf}
\usepackage{algorithm}
\usepackage{algorithmic}
\usepackage{pdfsync}
\usepackage{bm}
\usepackage{caption}
\newtheorem{thm}{Theorem}[]
\newtheorem{cor}{Corollary}
\newtheorem{lem}{Lemma}

\theoremstyle{remark}
\newtheorem{rem}{Remark}

\theoremstyle{definition}

\newcommand{\SWITCH}[1]{\STATE \textbf{switch} #1}
\newcommand{\ENDSWITCH}{\STATE \textbf{end switch}}
\newcommand{\CASE}[1]{\STATE \textbf{case} #1\textbf{:} \begin{ALC@g}}
	\newcommand{\ENDCASE}{\end{ALC@g}}

\newcommand{\DEFAULT}{\STATE \textbf{default:} \begin{ALC@g}}
	\newcommand{\ENDDEFAULT}{\end{ALC@g}}
\newcommand{\DEFAULTLINE}[1]{\STATE \textbf{default:} }

\begin{document}

\title{
{ Signaling Design for MIMO-NOMA with   \\Different Security Requirements}
}

\author{Yue Qi and Mojtaba Vaezi
}

\author{Yue Qi and Mojtaba Vaezi 
	\thanks{This paper was presented in part at IEEE Global Communications Conference \cite{qi2020power}.
			The authors are with the Department
		of Electrical and Computer Engineering, Villanova University, Villanova,
		PA 19085 USA (e-mail: yqi@villanova.edu; mvaezi@villanova.edu).}
}

\maketitle

\begin{abstract}
Signaling design for secure transmission in two-user multiple-input multiple-output (MIMO) non-orthogonal multiple access (NOMA) networks is investigated in this paper. The base station  broadcasts multicast data to all users and also  integrates additional services, unicast data targeted to certain users, and confidential data protected against eavesdroppers. We categorize the above MIMO-NOMA with  different security requirements into several communication scenarios. The associated problem in each scenario is nonconvex. We propose a unified approach, called the power splitting scheme, for optimizing the rate equations corresponding to the scenarios.  
The proposed method decomposes the optimization of the secure MIMO-NOMA channel into a set of simpler problems, including multicast, point-to-point, and wiretap MIMO problems,  corresponding to the  three basic messages: multicast, private/unicast, and confidential messages. We then leverage existing solutions to design signaling (covariance matrix) for the above problems such that the messages are transmitted with high security and reliability. Numerical results illustrate the efficacy of the proposed covariance matrix (linear precoding and power allocation) design in secure MIMO-NOMA transmission. The proposed method also outperforms existing solutions, when applicable.

In the case of   no multicast messages, we also reformulate the nonconvex problem into weighted sum rate (WSR) maximization problems by applying the block successive maximization method and generalizing the zero duality gap.
The two methods have their advantages and limitations. 
Power splitting is a general tool that can be applied to the MIMO-NOMA with any combination of the three messages (multicast, private, and confidential) whereas WSR maximization shows greater potential for secure MIMO-NOMA communication without multicasting. In such cases, WSR maximization provides a slightly better rate than the power splitting method.
\end{abstract}


\IEEEpeerreviewmaketitle

\section{Introduction}

%
  The unprecedented wave of emerging devices  has dramatically increased the requirements and challenges of resource allocation and spectrum utilization.  To fulfill 
  the demands, non-orthogonal multiple access (NOMA) at the physical (PHY) layer is a promising technique \cite{saito2013non,vaezi2019multiple}, and has attracted remarkable attention both in academia and industry.  
  For example, NOMA has also been proposed for inclusion in the
  Third Generation Partnership Project (3GPP) Long-Term Evolution
  Advanced (LTE-A) standard \cite{meredith2015study}. Therein, NOMA is referred to as multi-user superposition transmission (MUST) \cite{lee2016multiuser} built on LTE resource blocks. 


While several code-domain uplink NOMA schemes are developed in the literature \cite{vaezi2019multiple,cao2017resource}, downlink NOMA is based on well-known information-theoretic techniques for the broadcast channel (BC) \cite{weingartens2006capacity}. Then, in single-input single-output (SISO) NOMA networks, the superposition coding (SC) at the transmitter  and successive interference cancellation (SIC) at the receiver is the optimal strategy. 
Hence, a large body of work assumes NOMA to be equivalent to SC-SIC, 
and applies it to  multiple-input multiple-output  (MIMO) channels \cite{ding2015application, zeng2017sum, nguyen2017precoder, huang2018signal}.  
However, it is known that SC-SIC is not capacity-achieving in the MIMO-BC and \textit{dirty-paper coding} (DPC) is the only  optimal strategy \cite{weingarten2006capacity,vaezi2019noma,clerckx2021noma}. Similarly, in MIMO-NOMA with PHY layer security, SC-SIC cannot achieve secure capacity, and secret DPC (S-DPC) is the optimal solution \cite{liu2010multiple, Hung2013Liu}. 
In this paper, MIMO-NOMA is defined broadly and refers to any technique that allows simultaneous transmission over the same resources \cite{vaezi2019noma}, i.e., concurrent non-orthogonal transmission. That is, it is equivalent to the MIMO-BC.

\subsection{ MIMO-NOMA with Secrecy}

Today, there is a trend to
	merge multiple services in one transmission. This is referred to as \textit{PHY layer service integration} \cite{schaefer2014physical}. Integrated services  usually include three fundamental services: multicast, unicast, and confidential services,  which can be realized by common, private/individual, and confidential messages, respectively. Especially,  secure transmission of confidential messages requires PHY layer security which has been introduced as additional protection for secure transmission \cite{wyner1975wire}.
	
	This work is concerned with different security requirements for the two-user MIMO-NOMA networks, in which three  different types of messages can be transmitted:

\begin{itemize}
	\item \textit{Common message $M_0$ \cite{sidiropoulos2006transmit}:} a common message is transmitted in such a way that all users  can decode it. For example, the base station (BS) broadcasts daily news or amber alerts to all online users.
	\item \textit{Private message $M_p$ \cite{weingartens2006capacity}:
		a private or unicast/individual message is a message intended for a specific user}, since each user is only interested in its message. For instance, the BS provides targeted advertisements and recommended videos that are available only to  interested users. This message is not encoded securely, and as such, it can be decoded by other users.
	\item \textit{Confidential message $M_c$ \cite{khisti2010secure}:}  a confidential message is similar to a private message but is to be kept  secret from other users. For example, personal email accounts access and online banking transactions. Here, encoding is such that the message cannot be decoded by others. 
\end{itemize}




Early information-theoretic works\cite{weingarten2006capacity, ekrem2010gaussian,geng2014capacity,ly2008mimo,goldfeld2019mimo,Hung2013Liu} and have established the capacity region of  two-user MIMO-BC with different security requirements. 
This includes the MIMO-BC with one common and two independent private messages \cite{weingarten2006capacity}, \cite{ekrem2010gaussian,geng2014capacity}, the MIMO-BC with private, confidential, and common messages \cite{ly2008mimo, goldfeld2019mimo}, and the  
MIMO-BC with one common and two confidential messages \cite{ekrem2012capacity, Hung2013Liu}.
However, their primary purpose is
to derive capacity regions or to construct coding strategies that characterize certain rate regions. 
The solutions are based on DPC or S-DPC and usually are  given as a union over	all possible transmit covariance matrices satisfying certain	power constraints. Implementation of DPC requires sophisticated random coding  \cite{yoo2006optimality}, and finding practical dirty paper
codes close to the capacity is not easy \cite{lee2007high}. Linear precoding is a popular alternative to simplify the transmission design \cite{lee2007high, clerckx2021noma}.


The two-user MIMO-NOMA  can be classified  into three communication scenarios with different security requirements as shown in Fig~\ref{fig:scenarios}. The classification is mainly based on the well-established information-theoretic results:
	\begin{itemize}
		\item Scenario A (no security):  two independent private messages $M_{1p}$ and $M_{2p}$ (one for each user)  and a common message $M_{0}$ for both users are ordered \cite{weingarten2006capacity, ekrem2010gaussian,geng2014capacity}. In this case, we have a MIMO-NOMA with common and two private messages $(M_0, M_{1p}, M_{2p})$;  
		\item Scenario B (security for one user):   a confidential message $M_{1c}$ for user 1, a private message $M_{2p}$ for user 2, and one common message $M_{0}$ for both users are ordered \cite{goldfeld2019mimo}. Then, a MIMO-NOMA with common, private, and confidential messages $(M_0, M_{1c}, M_{2p})$ is formed;  
		\item Scenario C  (security for both users): two  confidential messages $M_{1c}$ and  $M_{2c}$ (one for each user), and a common message $M_{0}$ for both users are needed \cite{ekrem2012capacity, Hung2013Liu}. In this case, a MIMO-NOMA with common and confidential messages $(M_0, M_{1c}, M_{2c})$ is obtained.
	\end{itemize}  
	The three scenarios overall cover nine problems, or communication scenarios (see Table~\ref{tab_tasks}). The combinations of different types of messages are also named integrated services \cite{schaefer2014physical}.  
	

\begin{figure}[t]
	\centering
	\includegraphics[width=0.45\textwidth]{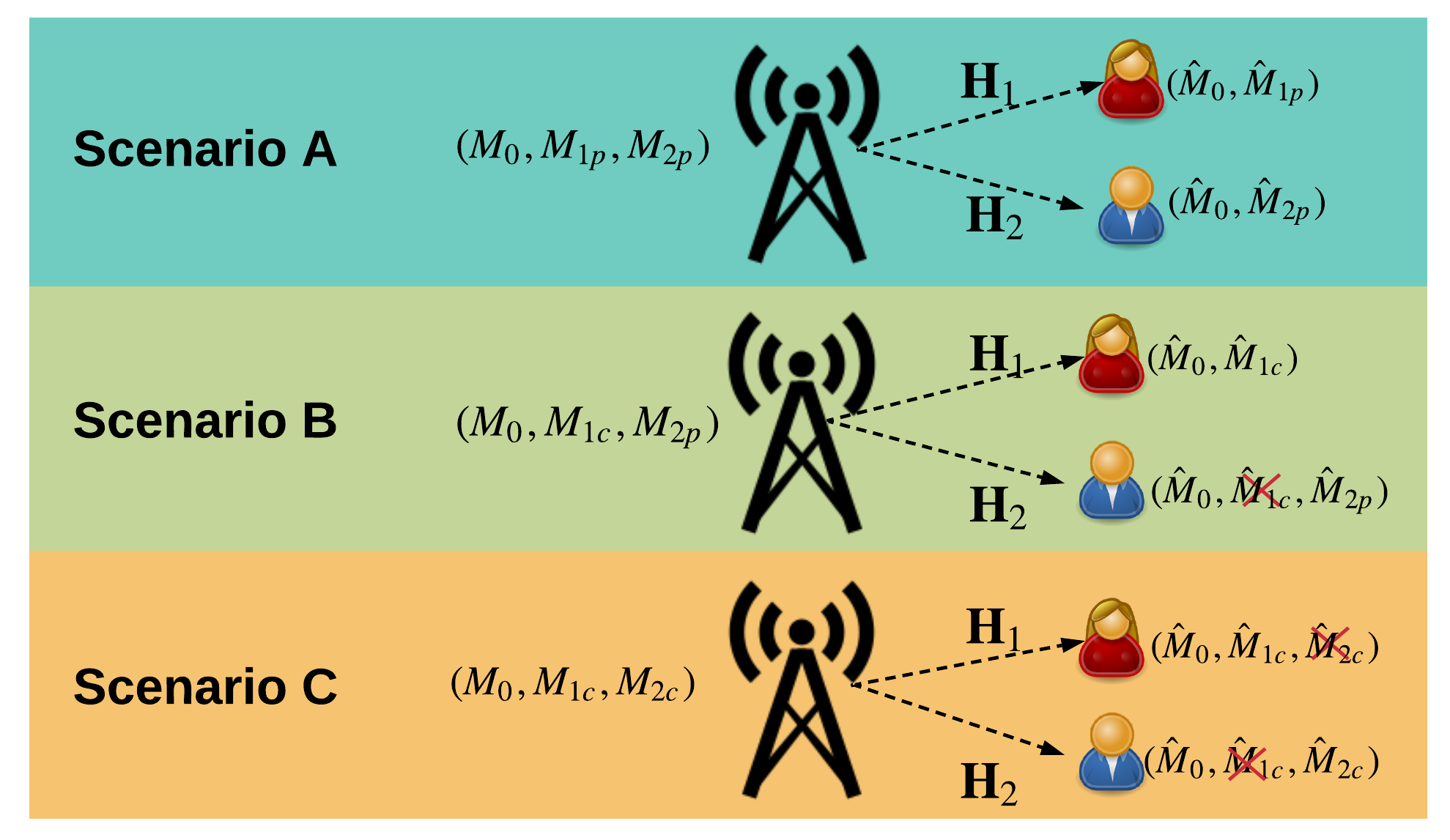}
	\caption{Communication scenarios with different combinations of security requirements based on the information-theoretic results. Consider three communication scenarios in which the base station sends different combinations of the three messages.
	}
	\label{fig:scenarios}
\end{figure}

\subsection{Motivation and Related Problems}

While the capacity regions of the three different messages are characterized,  it is still unknown how to identify optimal or implementation-efficient solutions.
Thus,
this paper is motivated by the following question:
How can we maximize the secrecy rate for the MIMO-NOMA with different security requirements in an acceptable computational complexity?

\begin{table*}[!tb]
	\caption{A summary of communication scenarios with different combinations of common, private, and confidential messages.}
	\label{tab_tasks}
	\centering
	\begin{tabular}{c |c c c c c c }
		\hline
		\rowcolor[HTML]{EFEFEF} 
		&  Communication scenarios                            &  $M_0$                    &  $M_1$                   &  $M_2$                    & Capacity region       &  Signaling schemes     \\ \hline
		\multirow{3}{*}{OMA}  & Multicasting  & Public & $-$ & $-$ & \cite{sidiropoulos2006transmit} & CAA \cite{zhu2012precoder}, closed-form \cite{qi2020power}                       \\ \cline{2-7} 
		& \cellcolor[HTML]{EFEFEF} P2P MIMO        &     \cellcolor[HTML]{EFEFEF} $-$   & \cellcolor[HTML]{EFEFEF}Private & \cellcolor[HTML]{EFEFEF}$-$  & \cellcolor[HTML]{EFEFEF}\cite{cover2012elements} & \cellcolor[HTML]{EFEFEF}SVD and WF \cite{cover2012elements}            \\ 
		& Wiretap  &   $-$  & Confidential & $-$ & \cite{khisti2010secure} & GSVD \cite{fakoorian2012optimal}, AOWF \cite{li2013transmit}, RM \cite{vaezi2017journal}   \\ \hline
		\multirow{7}{*}{NOMA} & \cellcolor[HTML]{EFEFEF}Two private    & \cellcolor[HTML]{EFEFEF}$-$   & \cellcolor[HTML]{EFEFEF}Private & \cellcolor[HTML]{EFEFEF}Private  & \cellcolor[HTML]{EFEFEF} \cite{vishwanath2003duality,weingarten2004capacity,weingartens2006capacity,viswanathan2003downlink} & \cellcolor[HTML]{EFEFEF}\color{black}GSVD \cite{chen2019asymptotic}, this work         \\ 
		& Private and confidential                            & $-$    & Confidential  & Private   & \cite{goldfeld2019mimo} & This work       \\ 
		& \cellcolor[HTML]{EFEFEF}Two confidential                                    &   \cellcolor[HTML]{EFEFEF}\color{black}$-$  & \cellcolor[HTML]{EFEFEF}\color{black}Confidential & \cellcolor[HTML]{EFEFEF}\color{black}Confidential& \cellcolor[HTML]{EFEFEF}\color{black}\cite{liu2010multiple} & \cellcolor[HTML]{EFEFEF} GSVD \cite{fakoorian2013optimality}, BSMM \cite{park2015weighted}, PS \cite{qi2020secure}         \\ 
		& Common and one confidential                        & Public &  \color{black}Confidential    & $-$   & \cite{HungD2010Liu} & GSVD \cite{weidong2016GSVD}, RM\cite{vaezi2019rotation}    \\ 
		& \cellcolor[HTML]{EFEFEF}Scenario A  &  \cellcolor[HTML]{EFEFEF}Public   &   \cellcolor[HTML]{EFEFEF}Private    & \cellcolor[HTML]{EFEFEF}Private   &   \cellcolor[HTML]{EFEFEF}\cite{weingarten2006capacity, ekrem2010gaussian,geng2014capacity} &  \cellcolor[HTML]{EFEFEF}This work \\ 
		& Scenario B  & Public    & Confidential  & Private   & \cite{ly2008mimo, goldfeld2019mimo} & This work \\ 
		& \cellcolor[HTML]{EFEFEF}Scenario C &  \cellcolor[HTML]{EFEFEF}Public  & \cellcolor[HTML]{EFEFEF}Confidential & \cellcolor[HTML]{EFEFEF}Confidential & \cellcolor[HTML]{EFEFEF}\cite{ekrem2012capacity, Hung2013Liu} & \cellcolor[HTML]{EFEFEF}This work  \\ \hline
	\end{tabular}
\end{table*}

The state-of-the-art includes solutions only for some special combinations of the messages and the orthogonal multiple access (OMA) case in which only one message, out of the three messages we mentioned earlier, is transmitted. These are summarized in Table~\ref{tab_tasks}, and some are highlighted below.

\begin{itemize}
	\item \textit{Two private messages \cite{vishwanath2003duality,weingarten2004capacity,weingartens2006capacity}:} When there is no common message in Scenario A, the problem reduces to the MIMO-BC and  
	DPC gives the capacity region. Alternatively, multiple access channel (MAC) to  BC  duality \cite{vishwanath2003duality} can be applied to iteratively achieve the capacity \cite{viswanathan2003downlink}. An analytic linear precoding scheme based on \textit{generalized singular value decomposition} (GSVD)--which  decomposes the MIMO channel into parallel SISO channels--is also designed for the special case where the two users are equipped with the same number of antennas \cite{chen2019asymptotic}.
	
	\item \textit{One confidential message \cite{khisti2010secure}:} When  there is neither a common nor private message in Scenario B, the system reduces to a MIMO wiretap 
	channel \cite{khisti2010secure}. Various linear precoding schemes, including  GSVD \cite{fakoorian2012optimal}, alternating optimization and water filling (AOWF) algorithm \cite{li2013transmit}, and rotation modeling (RM) \cite{vaezi2017journal} are known for this problem.
	
	\item \textit{Two confidential messages \cite{liu2010multiple}}: When $M_{0}$ is empty in Scenario C, the problem reduces to the  MIMO-BC with two confidential messages.  It is proven in \cite[Theorem 1]{liu2010multiple} that both users can reach their respective maximum	secrecy rates simultaneously by S-DPC under input covariance constraints.
	 Low-complexity approaches, such as  GSVD \cite{fakoorian2013optimality}, weighted sum rate (WSR) maximization with block successive maximization method (BSMM) \cite{park2015weighted}, and power splitting (PS) method \cite{qi2020secure} are developed. We generalize the PS into a more general case.
	
	\item \textit{Common and confidential messages \cite{HungD2010Liu}:} Different linear precoding schemes, including a GSVD-based precoding \cite{weidong2016GSVD} and RM with random search \cite{vaezi2019rotation} are proposed in this case.

	
	\item \textit{MIMO with only a common message \cite{sidiropoulos2006transmit, zhu2012precoder}:} If there is only a common message and no individual messages to be transmitted, the system becomes a multicast channel.   Heuristic precoding with iterations is investigated in \cite{zhu2012precoder}, analytical solutions with a convex tool  are in \cite{qi2020power}.
\end{itemize}
 However, these problems are all special cases of
	the three general scenarios mentioned earlier. Signaling designs (solutions) for the general cases are still unknown in general.  

\subsection{Contributions and Organization}
As illustrated, the problems listed in Table~\ref{tab_tasks} are all related to the three scenarios shown in Fig.~\ref{fig:scenarios}. Nonetheless, there are no solutions for the general cases. In this paper, we propose a new solution, named the power splitting method, which applies to all of those problems.  This method decomposes the secure MIMO-NOMA channel into point-to-point (P2P), wiretap, and multicasting MIMO channels. Then, we design one algorithm that can be used in all problems in Table~\ref{tab_tasks} to approach  their corresponding capacity regions.
The main contributions can be summarized as follows:
\begin{itemize}	
	\item We first split the total power among the three messages and then reformulate the secrecy capacity optimization problems into three sub-problems. Particularly, Scenario~A (only private massages) is decomposed into two P2P MIMO problems;  Scenario~B (private and confidential messages) is decomposed into one P2P MIMO problem and one wiretap channel; and, Scenario~C (only confidential messages) is decomposed into two wiretap channels.
	\item 	Linear precoder and power allocation matrices are designed for private and confidential messages by extending the analytical solution of the P2P MIMO problem and the numerical solution of wiretap channels to the MIMO-NOMA with different secrecy scenarios. For multicasting, we use a combination of analytical solutions and numerical solution based on a convex tool. Finally, we propose an algorithm for all  different secrecy scenarios.
	\item  When there is no common message, a WSR maximization  is formulated in all scenarios. We prove that  the WSR problem has zero duality gap in all scenarios, and the KKT conditions are necessary for the optimal solutions.  Besides, we derive and generalize an iterative algorithm for all  scenarios by applying the BSMM \cite{park2015weighted, razaviyaynunified}.  Especially, in Scenario~A, we provide an alternative solution that directly optimizes the4 WSR of the DPC region with BSMM instead of applying MAC-BC duality.
\end{itemize}

One main benefit of the proposed signaling design (power splitting, linear precoding, and power allocation) is its ability to be generalized to more complicated scenarios, e.g., when there are more than two users.

%

%
The remainder of this paper is organized as follows. In the
next section, we discuss the channel model and formulate the problems for the three  scenarios.
We introduce a power splitting method to all scenarios in Section~\ref{sec:IVA}, and a signaling design for each in Section~\ref{sec:IVB}. For the subcases without common messages, we generalize a WSR based on BSMM for all  scenarios in Section~\ref{sec:III}. 
We then present numerical results in  Section~\ref{sec:V} and conclude the paper in Section~\ref{sec:conclusion}.

\textit{Notations:} $\rm{tr}(\cdot)$ and $(\cdot)^{T}$ denote trace and transpose of matrices. $E\{\cdot\}$ denotes expectation. ${\rm{diag}}(\lambda_1, 
\dots, \lambda_n)$ represents diagonal matrix with elements  
$\lambda_1, 
\dots, \lambda_n$. $\mathbf{Q} \succcurlyeq \mathbf{0} $ represents that $\mathbf{Q}$ is a positive semidefinite matrix. $[x]^{+}$ gives the max value of $x$ and 0. $\mathbf{I}$ is the identity matrix.


\section{System Model} \label{Sec:II}

Considering a two-user MIMO-NOMA network, the BS equipped with $n_t$ antennas simultaneously serves two users, in which user~$1$ and user~$2$ are equipped with $n_1$ and  $n_2$ antennas, respectively. The transmitted signal to user~$1$ and user~$2$ share the same time slot and frequency.
\noindent The received signals at user~$1$ and user~$2$ are given by
\begin{subequations}\label{eq:signal model}
	\begin{align} 
	\mathbf{y}_1  = \mathbf{H}_1\mathbf{x} +\mathbf{w}_1,\\
	\mathbf{y}_2 = \mathbf{H}_2\mathbf{x} + \mathbf{w}_2,
	\end{align}
\end{subequations}
in which $\mathbf{H}_1 \in \mathbb{R}^{n_1 \times n_t}$ and $\mathbf{H}_2 
\in \mathbb{R}^{n_2 \times n_t}$ are the channel matrices for user~$1$ and 
user~$2$, respectively. $\mathbf{w}_1 \in \mathbb{R}^{n_1 \times 1}$ and 
$\mathbf{w}_2 \in \mathbb{R}^{n_2 \times 1}$ are independent identically 
distributed { (i.i.d.)} Gaussian noise vectors  whose elements are zero mean and 
unit variance.
The input $\mathbf{x} \in \mathbb{R}^{n_t \times 1}$ is a vector consisting of 
three components
\begin{align}
\mathbf{x} = \mathbf{x}_0 + \mathbf{x}_1 + \mathbf{x}_2,
\end{align}
where $\mathbf{x}_k \sim \mathcal{N}(\mathbf{0}, \mathbf{Q}_k)$, $k=0,1,2$,
is the input corresponding to two kinds of services: the multicast message $M_0$ and secure messages (private $M_p$ and/or confidential $M_c$)  of user~$1$ and 
user~$2$, in which $\mathbf{Q}_k\succcurlyeq \mathbf{0}$, is the 
covariance matrix corresponding to  $\mathbf{x}_k$. 
The channel input is subject to an average total power constraint 
	\begin{align} \label{eq:power cons}
	{\rm tr}(\mathbb{E}\{\mathbf{x}\mathbf{x}^{T}\}) = {\rm tr}(\mathbf{Q}_0+\mathbf{Q}_1+\mathbf{Q}_2) \leq P.
	\end{align}
We denote $R_{0}$, $R_{jp}$, and $R_{jc}$, $j = 1, 2$, as the achievable rates associated with the multicast, private, and confidential messages transmitted by the corresponding user $j$, respectively. 

%
	In the following, we provide the achievable rate region for each scenario.
	
\subsection{Scenario A (one common and two private messages)}
The DPC rate region of the MIMO-NOMA with common and two private messages is realized by \cite{weingarten2006capacity, ekrem2010gaussian},
\begin{align} \label{eq: level1rate}
	R_{{ A}}(P) = \rm{conv} \bigg\{\mathcal{R}^{\rm{DPC}}_{12} \cup \mathcal{R}^{\rm{DPC}}_{21} \bigg\}
\end{align}
in which $\rm{conv}$ is the convex hull operator. 
$\mathcal{R}^{\rm{DPC}}_{12}$   consists of all triplets $(R_{1p}, R_{2p},R_0)$ satisfying
\begin{subequations} \label{eq: mathmode1}
	\begin{align}
	&	R_0 \leq \min (R_{01}, R_{02}), \label{eq: mathmodel1_common} \\
	& R_{1p} \leq \frac{1}{2} \log|\mathbf{I} +\mathbf{H}_1 \mathbf{Q}_1 
	\mathbf{H}_1^T|, \label{eq: mathmodel1_confi1}  \\
	&	R_{2p} \leq  \frac{1}{2}\log|\mathbf{I} + ({\mathbf{I} +\mathbf{H}_2 \mathbf{Q}_1 
		\mathbf{H}_2^T})^{-1}\mathbf{H}_2 
		\mathbf{Q}_2 \mathbf{H}_2^T | \label{eq:  mathmodel1_confi2} 
	\end{align}
\end{subequations}
where
\begin{align} \label{eq:r0details}
 R_{0j} \triangleq \frac{1}{2} \log \bigg|\mathbf{I} +\frac{ \mathbf{H}_j 
	\mathbf{Q}_0 \mathbf{H}_j^T}{(\mathbf{I}+\mathbf{H}_j 
	(\mathbf{Q}_1+\mathbf{Q}_2) \mathbf{H}_j^T)}\bigg|, j=1, 2
\end{align}
with the total power constraint \eqref{eq:power cons}.
$\mathcal{R}^{\rm{DPC}}_{21}$ is obtained from $\mathcal{R}^{\rm{DPC}}_{12}$ by swapping the subscripts 1 and 2 corresponding to different DPC encoding orders. 
 Specifically,  when each of
	the users has a single antenna, the problem  can be transferred to a linear semi-definite convex optimization \cite[Section III]{weingarten2006capacity}, but the MIMO case is in general still unknown.
Without the common message, the capacity of the MIMO BC is given in \cite{weingarten2004capacity, vishwanath2003duality}.

\subsection{Scenario B (common, private, and confidential messages)}
In this scenario, only user 1 requires a confidential message. The secrecy capacity region $R_{{ B}}(P)$ under a total power constraint \eqref{eq:power cons} is given by a set of rate triples $(R_{1c}, R_{2p},R_0)$ satisfying \cite[Theorem 2]{goldfeld2019mimo}
\begin{subequations} \label{eq: mathmode2}
	\begin{align}
	&	R_0 \leq \min (R_{01}, R_{02}), \label{eq: mathmodel2_common} \\
	& R_{1c} \leq \frac{1}{2} \log|\mathbf{I} +\mathbf{H}_1 \mathbf{Q}_1 
	\mathbf{H}_1^T|
	-\frac{1}{2} \log|\mathbf{I} +\mathbf{H}_2 \mathbf{Q}_1 
	\mathbf{H}_2^T|,\label{eq: mathmodel2_confi1}  \\
	& R_{2p} \leq  \frac{1}{2}\log|\mathbf{I} + ({\mathbf{I} +\mathbf{H}_2 \mathbf{Q}_1 	\mathbf{H}_2^T})^{-1}\mathbf{H}_2 \mathbf{Q}_2 \mathbf{H}_2^T |.  \label{eq:  mathmodel2_confi2}
	\end{align}
\end{subequations}
The entire secrecy capacity region is achieved using DPC to cancel out the signal of private $M_{2p}$ at user 2, the other variant, i.e., DPC against $M_{1c}$ at user 1, is unnecessary. This is different from  Scenario~A for which the capacity region is exhausted by taking the convex hull of both variants ($\mathcal{R}^{\rm{DPC}}_{21}$ and $\mathcal{R}^{\rm{DPC}}_{12}$).

\subsection{Scenario C (one common and two confidential messages)}
The secrecy capacity region $R_{{ C}}(P)$ of the MIMO-BC with one common and two confidential messages under the average total power constraint \eqref{eq:power cons}
can be expressed as \cite[Theorem 2]{Hung2013Liu}
\begin{subequations} \label{eq: mathmodel3}
\begin{align}
	&R_0 \leq \min (R_{01}, R_{02}), \label{eq: mathmodel3_common} \\ 
&R_{1c} \leq  \frac{1}{2}\log|\mathbf{I} + { \mathbf{H}_1 \mathbf{Q}_1 \mathbf{H}_1^T}|-\frac{1}{2}\log|\mathbf{I} +  \mathbf{H}_2 \mathbf{Q}_1 \mathbf{H}_2^T|, \label{eq: mathmodel3_confi1} \\
&R_{2c} \leq   \frac{1}{2}\log|\mathbf{I} + ({\mathbf{I} +\mathbf{H}_2 \mathbf{Q}_1 
	\mathbf{H}_2^T})^{-1}\mathbf{H}_2 
\mathbf{Q}_2 \mathbf{H}_2^T |  \notag\\
& \quad \quad \quad -\frac{1}{2}\log |\mathbf{I} + ({\mathbf{I} +\mathbf{H}_1 \mathbf{Q}_1 \mathbf{H}_1^T})^{-1}{\mathbf{H}_1 \mathbf{Q}_2 \mathbf{H}_1^T}|,  \label{eq: mathmodel3_confi2}
\end{align}
\end{subequations}
The secrecy capacity region is characterized by S-DPC \cite{Hung2013Liu,ekrem2012capacity} \footnote{The S-DPC can assure security between the two users because a precoding matrix is selected such that it satisfies two goals \cite[Remark 5]{liu2010multiple}.  First,  it helps to cancel the precoding signal representing message $M_{2c}$, so that $M_{1c}$ can be served with an interference-free legitimate user channel. Second, it boosts the secrecy for  message $M_{2c}$ by causing interference (artificial noise) to user~1. In other words, user~1 can remove the interference of user~2 but is not able to decode the message of user~2.}. 
In this scenario, both of the users are legitimate and secret from each other. User 1 with confidential messages $M_{1c}$ treats user 2 as an eavesdropper, and vice versa.

The secrecy capacity regions in \eqref{eq: mathmode1}, \eqref{eq: mathmode2}, \eqref{eq: mathmodel3} are obtained via DPC or S-DPC which can be obtained by an exhaustive search over the set of 
all possible input covariance matrices. However, the complexity of such  
methods for practical implementations is prohibitive, which motivates us to develop a simpler 
 signaling scheme. The covariance matrices achieving the capacity regions are not known in general due to the non-convexity.

\section{Power Splitting Method for MIMO-NOMA  in All Scenarios}\label{sec:IV}
 In order to introduce a new simpler and faster solution, we split the total power for three messages in each scenario. Then, we decouple the MIMO-NOMA channel of all secrecy scenarios into three different problems and solve them separately. 
\subsection{Decomposing Secure MIMO-NOMA  into Simpler Channels} \label{sec:IVA}

We decompose MIMO-NOMA into different problems in this section. Due to some overlapping, such as the privacy part in  Scenario~A and Scenario~B, the confidentiality part in Scenario~B and  Scenario~C, we start with \textit{Step~1} to split the power between user~1 and user~2 for different usages; Then, \textit{Step~2a} and \textit{Step~2b} are for user~1 with private messages in  Scenario A and confidentiality in  Scenario~B,  respectively; \textit{Step~3a} and \textit{Step~3b} are for user~2 with private messages in  Scenario~A,  and confidentiality in  Scenario~C, correspondingly; Last, \textit{Step 4} is for common message in all  scenarios.

\subsubsection*{\textbf{Step 1}}
Introducing power splitting 
factor $\alpha_k \in [0,1]$, and $\sum_{k}\alpha_k=1$, $k=0,1,2$. We dedicate a fraction $\alpha_1$ of 
the total power to user~$1$ for the private message $M_{1p}$ or the confidential message $M_{1c}$ ($P_1= \alpha_1 P$), and fraction $\alpha_2 \in [0,\alpha_1]$ to  user~$2$ for the secure message $M_{2p}$ or $M_{2c}$ ($P_2= \alpha_2 P$). Last, we allocate the remaining power  to the common message $M_0$ for both users ($P_0=\alpha_0 P = (1-\alpha_1-\alpha_2)P$) \footnote{The optimal solution using total power throughout the paper.}.

\subsubsection*{\textbf{Step 2a}}
We design the secure precoding for user~$1$ with private message $M_{1p}$ in \eqref{eq: mathmodel1_confi1} for Scenario A. 
Since the rate $R_{1p}(\alpha_1)$ of  user 1  is only controlled by the covariance matrix $\mathbf{Q}_1$ under the power constraint $P_1$, the interference-free link can be seen as a P2P MIMO with power $P_1$, which is
\begin{subequations}\label{eq:R1level1}
	\begin{align}
	R_{1p}(\alpha_1) &= \max \limits_{\mathbf{Q}_1 \succeq \mathbf{0}} \frac{1}{2}\log|\mathbf{I} +\mathbf{H}_1 \mathbf{Q}_1 \mathbf{H}_1^T| \\
	&{\rm s.t.}\quad {\rm tr}(\mathbf{Q}_1)\leq P_1 = \alpha_1 P. \label{eq:constrainr11}
	\end{align}
\end{subequations}
The solution for $\mathbf{Q}_1^*$ has been well-developed analytically through singular value decomposition (SVD) and water filling (WF) \cite{cover2012elements}.

\subsubsection*{\textbf{{Step 2b}}} 
 In Scenario B and  Scenario C, we design secure precoding for user~$1$ with confidential messages $M_{1c}$ while treating  the second user as an eavesdropper. Because covariance matrix $\mathbf{Q}_1$ is the only variable in \eqref{eq: mathmodel2_confi1} and \eqref{eq: mathmodel3_confi1}, the problem can be 
seen as a wiretap channel under a transmit power $P_1$, which 
is 
\begin{subequations}\label{eq:R1}
	\begin{align}
	R_{1c}(\alpha_1) &= \max \limits_{\mathbf{Q}_1 \succeq 
		\mathbf{0}} 
	\frac{1}{2}\log \frac{| \mathbf{I} + { \mathbf{H}_1 \mathbf{Q}_1 
			\mathbf{H}_1^T}|}{ |\mathbf{I} +  
		\mathbf{H}_2 \mathbf{Q}_1 \mathbf{H}_2^T|}, \label{eq:R1C1} \\
	&{\rm s.t.}\quad {\rm tr}(\mathbf{Q}_1)\leq P_1 = \alpha P. \label{eq:R1C2}
	\end{align}
\end{subequations}
This problem is 
now the well-known MIMO wiretap channel \cite{vaezi2017journal}, and standard MIMO wiretap solutions can be applied to obtain  $\mathbf{Q}_1^*$. 

\subsubsection*{\textbf{{Step 3a}}}
We design secure precoding for user 2 to maximize the secrecy rate $R_{2p}(\alpha_2)$ of user~$2$ with private message $M_{2p}$, we apply $\mathbf{Q}_1^*$ obtained in \textit{Step 2a} and \textit{Step 2b} to \eqref{eq: mathmodel1_confi2} and \eqref{eq:  mathmodel2_confi2} for { Scenario A} and { Scenario B}, respectively. The interference channel can be transformed to an interference-free link over modified channels. Thus, \eqref{eq: mathmodel1_confi2} or \eqref{eq: mathmodel2_confi2} can be represented as
\begin{subequations} \label{rate2p}
	\begin{align}
	R_{2p}(\alpha_2) &= \max \limits_{\mathbf{Q}_2 \succeq \mathbf{0}}
	\frac{1}{2} \log|\mathbf{I} + ({\mathbf{I} +\mathbf{H}_2 \mathbf{Q}^{*}_1 
		\mathbf{H}_2^T})^{-1}\mathbf{H}_2 
	\mathbf{Q}_2 \mathbf{H}_2^T | \label{eq:constrainr1} \\
	&{\rm s.t.}\quad {\rm tr}(\mathbf{Q}_2)\leq P_2 = \alpha_2 P. \label{eq:constrainr12}
	\end{align}
\end{subequations}
Since $\mathbf{Q}^{*}_1 $ is given after solving \textit{decomposing MIMO-NOMA into different problems}, in the following we show
that the above problem can be seen as a P2P MIMO problem under power $P_2$.
\begin{thm} \label{thm1}
The optimization problem in \eqref{rate2p} with interference
from user 1 can be converted to the optimization of
a standard P2P MIMO  channel
	\begin{align}
	\dot{\mathbf{H}}_2 \triangleq 
	\mathbf{B}^{-\frac{1}{2}}\mathbf{C}^{T}\mathbf{H}_2
	\end{align}
	for user~$2$, in which $\mathbf{B}$ and  $\mathbf{C}$ are the eigenvalue and eigenvector of ${\mathbf{I} +\mathbf{H}_2 \mathbf{Q}^{*}_1\mathbf{H}_2^T}$.
\end{thm}
\begin{proof}
	Define  
	\begin{align}\label{eq: sigma1}
	\mathbf{\Sigma} \triangleq \mathbf{I}+\mathbf{H}_2
	\mathbf{Q}^{*}_1
	\mathbf{H}_2^T=\mathbf{C}\mathbf{B}\mathbf{C}^T.  
	\end{align}
	Then, the secrecy rate for user~$2$ with private $M_{2p}$ can be written as
	\begin{align} \label{eq: derivation1}
	R_{2p}(\alpha_2) &=\max \limits_{\mathbf{Q}_2 \succeq \mathbf{0}}  
	\frac{1}{2}\log {|\mathbf{I} + \mathbf{\Sigma}^{-1} \mathbf{H}_2 \mathbf{Q}_2 
		\mathbf{H}_2^T|} \notag \\
	& = \max \limits_{\mathbf{Q}_2 \succeq \mathbf{0}}  \frac{1}{2}\log 
	{|\mathbf{I} +  \mathbf{C}\mathbf{B}^{-1}\mathbf{C}^{T} \mathbf{H}_2 \mathbf{Q}_2 
		\mathbf{H}_2^T|}  \notag \\
	&\stackrel{(a)}{=} \max \limits_{\mathbf{Q}_2 \succeq \mathbf{0}}  \frac{1}{2}\log 
	{|\mathbf{I} +  
		\mathbf{B}^{-\frac{1}{2}}\mathbf{C}^{T}\mathbf{H}_2 
		\mathbf{Q}_2 
		\mathbf{H}_2^T\mathbf{C}\mathbf{B}^{-\frac{1}{2}}|} 
	\notag \\
	&=\max \limits_{\mathbf{Q}_2 \succeq \mathbf{0}}  \frac{1}{2}\log 
	{|\mathbf{I} +  \dot{\mathbf{H}}_2 \mathbf{Q}_2 
		\dot{\mathbf{H}}^T_2|},  
	\end{align}
	in which $(a)$  holds because of Sylvester’s determinant theorem, i.e., $\det(\mathbf{I} + \mathbf{X}\mathbf{Y}) = 
	\det(\mathbf{I} + \mathbf{Y}\mathbf{X})$, 
	$\mathbf{C}$ is orthogonal, i.e., $\mathbf{C}^{-1} = \mathbf{C}^{T}$, and $\mathbf{B}$  is a diagonal matrix.
\end{proof}
In view of \eqref{eq: derivation1}, the problem in \eqref{rate2p} becomes the standard P2P
MIMO without interference over a modified channel. The solution  $\mathbf{Q}_2^*$ can be obtained the same as \textit{{Step 2a}}.

\subsubsection*{\textbf{{Step 3b}}}
We design secure precoding for user 2  to maximize the secrecy rate $R_{2c}(\alpha_2)$ of user~$2$ with confidential message $M_{2c}$, we apply $\mathbf{Q}_1^*$ obtained in \textit{Step 2b} to \eqref{eq: mathmodel3_confi2} for Scenario~C. Thus, \eqref{eq: mathmodel3_confi2} can be represented as
\begin{subequations} \label{eq:R2}
	\begin{align}
	R_{2c}({\alpha_2}) &= \max \limits_{\mathbf{Q}_2 \succeq \mathbf{0}} 
	\bigg\{\frac{1}{2}{ \log}\bigg|\mathbf{I} + \frac{\mathbf{H}_2 \mathbf{Q}_2 
		\mathbf{H}_2^T}{\mathbf{I} +\mathbf{H}_2 \mathbf{Q}^{*}_1 
		\mathbf{H}_2^T}\bigg|\notag \\
	& \quad \quad \quad \quad \quad \quad	-\frac{1}{2} \log \bigg|\mathbf{I} + \frac{\mathbf{H}_1 \mathbf{Q}_2 \mathbf{H}_1^T}{\mathbf{I} +\mathbf{H}_1 \mathbf{Q}^{*}_1 \mathbf{H}_1^T}\bigg|\bigg\}, \label{eq:R2C1}\\
	&{\rm s.t.}\quad {\rm tr}(\mathbf{Q}_2)\leq P_2 = (1-\alpha) P. \label{eq:R2C2}
	\end{align}
\end{subequations} 
Since  $\mathbf{Q}^{*}_1$  is given after solving \eqref{eq:R1}, next we show 
that the  problem \eqref{eq:R2} can be seen as a wiretap 
channel where users~$2$ and $1$ 
are the legitimate user and eavesdropper, respectively. 
\begin{thm}{\cite{qi2020secure}}\label{thm:theorem1}
	The above channel can be converted to a standard MIMO wiretap channel with 	\begin{subequations} \label{eq: sigmas1}
		\begin{align}
		\ddot{\mathbf{H}}_1 \triangleq 
		\mathbf{D}^{-\frac{1}{2}}_a\mathbf{E}^{T}_a\mathbf{H}_1, \\
		\ddot{\mathbf{H}}_2 \triangleq 
		\mathbf{D}^{-\frac{1}{2}}_b\mathbf{E}^{T}_b\mathbf{H}_2,
		\end{align}
	\end{subequations}
	in which $\mathbf{D}_a$ and 
	$\mathbf{E}_a$ are the eigenvalues and eigenvectors  of $\mathbf{I}+\mathbf{H}_1 
	\mathbf{Q}^{*}_1 \mathbf{H}_1^T$, and $\mathbf{D}_b$ and 
	$\mathbf{E}_b$ are the eigenvalues and eigenvectors of 
	$\mathbf{I}+\mathbf{H}_2 
	\mathbf{Q}^{*}_1 \mathbf{H}_2^T$.
\end{thm}
\noindent Then, the rate for user~$2$ can be written as
\begin{align} \label{eq: derivation}
	R_{2c}(\alpha_2)
	&=\max \limits_{\mathbf{Q}_2 \succeq \mathbf{0}}  \frac{1}{2}\log 
	\frac{|\mathbf{I} +  \ddot{\mathbf{H}}_2 \mathbf{Q}_2 
		\ddot{\mathbf{H}}^T_2|}{
		|\mathbf{I} +  \ddot{\mathbf{H}}_1 \mathbf{Q}_2 
		\ddot{\mathbf{H}}^T_1|},  
\end{align}

In view of  \eqref{eq: derivation}, it is seen that similar to \eqref{eq:R1C1},  \eqref{eq:R2C1} is the rate for a MIMO wiretap channel with  channels $\ddot{\mathbf{H}}_2$ for the legitimate user and $\ddot{\mathbf{H}}_1$ for the eavesdropper.  This problem now transfers to a MIMO wiretap channel, and we can obtain $\mathbf{Q}_2^*$ using any standard MIMO wiretap solutions.

\subsubsection*{\textbf{{Step 4}}}
After distributing the power to both users for secrecy messages, we allocate the remaining power  $P_0 = \alpha_0 P, \alpha_0 = 1- \alpha_1-\alpha_2$ to  the common message $M_0$ for both users.
The equation \eqref{eq: mathmodel1_common}, \eqref{eq: mathmodel2_common}, and \eqref{eq: mathmodel3_common} becomes
\begin{subequations} \label{eq:r0}
	\begin{align} 
	R_0(\alpha_0) &= \max  \limits_{\mathbf{Q}_0 \succeq \mathbf{0}} \min \{ R_{0j}\}, j=1, 2 \\
	&{\rm s.t.}\quad {\rm tr}(\mathbf{Q}_0)\leq P_0 = \alpha_0 P,
	\end{align}
\end{subequations}
Since  $\mathbf{Q}^{*}_1$ and $\mathbf{Q}^{*}_2$ are given, we can show that the above problem  becomes MIMO multicasting \cite{sidiropoulos2006transmit} by applying the same approach as  Theorem~\ref{thm1} again into \eqref{eq:r0details}.
Specifically, let us define the denominator of \eqref{eq:r0details} as
\begin{align} 
\mathbf{K}_j \triangleq \mathbf{I}+\mathbf{H}_j 
(\mathbf{Q}^{*}_1+\mathbf{Q}^{*}_2) 
\mathbf{H}_j^T   \triangleq \mathbf{F}_j\mathbf{G}_j\mathbf{F}_j^T,  \label{eq:psi}
\end{align}
for ${j=1,2}$, where the second equality is given by eigenvalue decomposition. Then, $ R_{0j} $ can be rewritten as

\begin{align} \label{eq:R0j}
R_{0j}=& \frac{1}{2} \log 
\big|\mathbf{I}+{\mathbf{K}_j^{-1}}\mathbf{H}_j\mathbf{Q}_{0}\mathbf{H}_j^{T}\big|,  \notag\\
=& \frac{1}{2} \log 
\big|\mathbf{I}+\mathbf{G}_j^{-\frac{1}{2}}\mathbf{F}_j^{T}\mathbf{H}_j\mathbf{Q}_{0}\mathbf{H}_j^{T}\mathbf{F}_j\mathbf{G}_j^{-\frac{1}{2}}
\big|, \notag \\
=& \frac{1}{2}\log 
|\mathbf{I} +  \dddot{\mathbf{H}}_j \mathbf{Q}_0 
\dddot{\mathbf{H}}_j^T|.   
\end{align}
See the proof in Theorem~\ref{thm1}. Then we have  
$\dddot{\mathbf{H}}_j = 
\mathbf{G}_j^{-\frac{1}{2}}\mathbf{F}_j^{T}\mathbf{H}_j$, $j=1,2$.

The problem \eqref{eq:r0} with \eqref{eq:R0j} is now identified as the MIMO multicasting which is to maximize the minimum user rate configuration, and the optimal solution $\mathbf{Q}^{*}_{0}$ can be achieved by 
semidefinite programming (SDP), i.e., \texttt{CVX}, however, it may incur a high computational complexity for multiple users and antennas. As we will see in the next subsection,  analytical solutions together with a convex tool for different cases are proposed for multicast transmission.	

\subsection{{ The Signaling Design}} \label{sec:IVB}

\begin{figure}[t]
	\centering
	\includegraphics[width=0.48\textwidth]{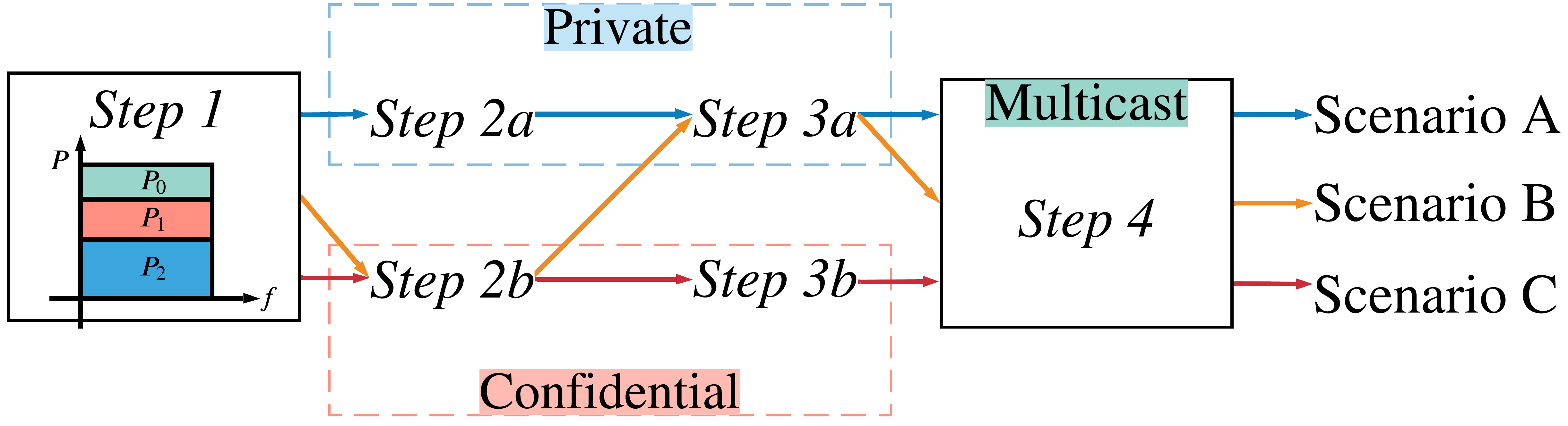}
	\caption{System structure of power splitting method for different security scenarios.}
	\label{fig:structure}
\end{figure}

	The structure of our decomposition of the MIMO-NOMA into different problems is shown in Fig.~\ref{fig:structure}. Then, we solve each sub-problem in this subsection, i.e., design precoding and power allocation for all the secrecy  scenarios. Scenario~A is composed of two P2P MIMO and one multicasting;  Scenario~B consists of one wiretap channel, one P2P MIMO, and one multicasting;  Scenario~C has two wiretap channels and one multicasting.

\subsubsection*{Scenario A}
(\textit{Step~1 $\rightarrow$ Step~2a $\rightarrow$  Step~3a $\rightarrow$ Step~4})

Problem \eqref{eq:R1level1} is a P2P MIMO problem which is convex and has a closed-form solution given in the following Lemma \cite{cover2012elements}.
\begin{lem} \cite{cover2012elements} \label{waterfilling}
	For P2P MIMO problem  $\max_{\mathbf{Q}\succcurlyeq \mathbf{0}}  \log|\mathbf{I}+\mathbf{H}\mathbf{Q}\mathbf{H}^{T}|$ under a total power constraint, the optimal solution is given by $\mathbf{Q}^{*} = \mathbf{\Psi}\mathbf{\Gamma}\mathbf{\Psi}^{T}$.
	in which  $\mathbf{H}=\mathbf{\Phi}{\rm diag}(\tau_1, \tau_2,\cdots,\tau_n)\mathbf{\Psi}^{T}$, $\tau_i \geq 0$, $\forall i$, $\mathbf{\Gamma}= {\rm diag}[(\mu-1/\tau_1^2)^{+}, \cdots, (\mu-1/\tau_n^2)^{+}]$, $\mu$ is the water level.
\end{lem}
\noindent The solutions of \eqref{eq:R1level1} in \textit{Step~2a}  and  \eqref{rate2p} in \textit{Step~3a} are achieved by replacing  $\mathbf{H}$ in Lemma~\ref{waterfilling} by $\mathbf{H}_1$  and $\dot{\mathbf{H}}_2$, respectively, using Theorem~\ref{thm1}.

To precode for the common message $M_0$ in \textit{Step 4}.  Define the optimal precoding matrices $\mathbf{Q}_{01}^{*}$ and $\mathbf{Q}_{02}^{*}$ for $R_{01}$  and  $R_{02} $ in \eqref{eq:R0j}, respectively, then we have \cite{qi2020power}
\begin{itemize}
	\item \textit{Case 1:}  $ R_{01}(\mathbf{Q}_{01}^{*})\leq 
	R_{02}(\mathbf{Q}_{01}^{*})$, then 
	the optimal 
	multicast covariance matrix of \eqref{eq:r0} is $\mathbf{Q}^{*}_0:= 
	\mathbf{Q}_{01}^{*} $. 
	\item \textit{Case 2:}  $ R_{01}(\mathbf{Q}_{02}^{*})\geq 
	R_{02}(\mathbf{Q}_{02}^{*})$, 	the optimal 
	multicast covariance matrix of \eqref{eq:r0} is $\mathbf{Q}^{*}_0:= 
	\mathbf{Q}_{02}^{*} $. 
	\item \textit{Case 3:} Otherwise, the optimal 
	multicast covariance matrix of \eqref{eq:r0} can be obtained by a random search.
\end{itemize}

For \textit{Case 1} or \textit{Case 2}, Lemma~\ref{waterfilling} \cite{cover2012elements} is  applied.
 For \textit{Case~3}, optimal $\mathbf{Q}^{*}_0$ happens 
when 
the two convex functions are equal. Then, we can generate $\mathbf{Q}_0$ using the rotation method and search the parameters non-linearly \cite{vaezi2019rotation}.

Finally, DPC rate region $\mathcal{R}^{\rm{DPC}}_{21}$ can be reached by exhaustively searching over all power fractions $\alpha_1$, $\alpha_2$ and  $\alpha_0$. For each pair of power splitting parameters $\alpha_1$, $\alpha_2$, and $\alpha_0$, we solve  precoding matrices $\mathbf{Q}^{*}_1$, $\mathbf{Q}^{*}_2$, and $\mathbf{Q}^{*}_0$ (and thus $R_{1p}(\alpha_1)$, $R_{2p}(\alpha_2)$, and $R_{0}(\alpha_0)$). Alternatively, $\mathcal{R}^{\rm{DPC}}_{21}$ is obtained by encoding the private messages for user 2 and user 1, then the common message for both. We can solve $\mathbf{Q}^{*}_2$ followed by $\mathbf{Q}^{*}_1$ and $\mathbf{Q}^{*}_0$ to obtain $\bar{R}_{1p}(\alpha_1)$, $\bar{R}_{2p}(\alpha_2)$, and $\bar{R}_{0}(\alpha_0)$, respectively.

\begin{cor} \label{cor:corollary1}
	The achievable DPC rate region for secure MIMO-NOMA Scenario A under the total power is the convex hull of all rate triplets
	\begin{align}
	R_{{ A}}(P) = 
 {\rm conv} \bigg\{ \bigg(\bigcup_{\alpha_k} \mathcal{R}^{{\rm DPC}}_{12}(\alpha_k)  \bigg) \bigcup \bigg(  \bigcup_{\alpha_k} \mathcal{R}^{{\rm DPC}}_{21}(\alpha_k)\bigg) \bigg\},
	\end{align}
\end{cor}
\noindent $\mathcal{R}^{\rm{DPC}}_{12}(\alpha_k) = (R^{*}_{1p}(\alpha_1),  R^{*}_{2p}(\alpha_2), R^{*}_0(\alpha_0))$, $k=0,1,2$, and is obtained by encoding the private messages for first user~1 then user~2 followed by the common message for both, whereas $\mathcal{R}^{\rm{DPC}}_{21}(\alpha_k) = (\bar{R}^{*}_{1p}(\alpha_1), \bar{R}^{*}_{2p}(\alpha_2), \bar{R}^{*}_0(\alpha_0))$ is obtained in the reverse order of private messages (first user~2 then user~1). 

\subsubsection*{Scenario B}(\textit{Step~1 $\rightarrow$ Step~2b $\rightarrow$ Step~3a $\rightarrow$ Step~4})

In light of our decomposition in the previous section, Scenario~B is characterized as one MIMO wiretap channel in \textit{Step~2b} and one P2P MIMO problem in \textit{Step~3a}.
Standard MIMO wiretap solutions can be applied to design 
covariance matrix $\mathbf{Q}_1$.  
One fast approach is the rotation-based linear precoding \cite{vaezi2017journal}. In this method, the covariance matrix 
$\mathbf{Q}_1$ 
is eigendecomposed into one rotation matrix $\mathbf{V}_1$ and one power 
allocation matrix $\mathbf{\Lambda}_1$ \cite{vaezi2017journal, vaezi2019rotation} 
as
\begin{align}\label{eq_eig1}
\mathbf{Q}_1= \mathbf{V}_1 \mathbf{\Lambda}_1 \mathbf{V}_1^T.
\end{align}
Consequently, the secrecy capacity of 
user~1 is \footnote{This paper is focused on the two-user case. For  $K$-user MIMO wiretap channel with an external eavesdropper,  $K$ covariance matrices should be constructed independently. When applying the rotation method for user $k$, we will have independent rotation parameters for each user to optimize.  For multiple external eavesdroppers, a large-scale analysis may be needed to evaluate the security performance \cite{liu2017enhancing}.}
\begin{subequations} \label{eq:r1star}
	\begin{align} 	
	R_{1c}(\alpha_1)
	&=\max \limits_{\mathbf{Q}_1 \succeq \mathbf{0}} \frac{1}{2}\log \frac{|\mathbf{I} 
		+ { \mathbf{H}_1 \mathbf{V}_1\mathbf{\Lambda}_1\mathbf{V}_1^{T} 
			\mathbf{H}_1^T}|}{|  \mathbf{I} 
		+ \mathbf{H}_2 
		\mathbf{V}_1\mathbf{\Lambda}_1\mathbf{V}_1^{T} \mathbf{H}_2^T|}, 
	\label{eq:r1star1}\\
	&{\rm s.t.}\quad \sum_{n = 1}^{n_t} \lambda_{1n} \leq P_1 = 
	\alpha P, \label{eq:r1starconst}
	\end{align}
\end{subequations}
in which $\lambda_{1n}$, $n= \{1,\dots,n_t\}$, is a diagonal 
element of matrix $\mathbf{\Lambda}_1 = {\rm diag}(\lambda_{11}, \dots, \lambda_{1n_t})$.
The rotation matrix $\mathbf{V}_1$ can be obtained by
\begin{align}\label{eq_Vnbyn_}
\mathbf{V}_1=\prod_{p=1}^{n_t-1}\prod_{q=p+1}^{n_t} \mathbf{V}_{pq},
\end{align}
in which the basic rotation matrix $\mathbf{V}_{pq}$ is a Givens matrix which is an identity matrix except that its elements in the 
$p$th row and $q$th column, i.e., $v_{pp}$, $v_{pq}$, 
$v_{qp}$, and $v_{qq}$ are replaced by 
\begin{align}\label{eq_VnDsub}
\left[
\begin{matrix}
v_{pp}	&v_{pq}\\
v_{qp}	&v_{qq}
\end{matrix}
\right]
=\left[
\begin{matrix}
\cos\theta_{1pq}	&-\sin \theta_{1pq}\\
\sin\theta_{1pq}	&\cos \theta_{1pq}
\end{matrix}
\right],
\end{align}
in which  
$\theta_{1pq}$ is rotation angle corresponding to the rotation 
matrix $ \mathbf{V}_{pq} $.
Then, we will optimize the parameterized problem 
by applying numerical approaches such as Broyden-Fletcher-Goldfarb-Shanno (BFGS) method \cite{nocedal2006numerical} to obtain the solution $\mathbf{Q}^{*}_1$ (thus ${R}^{*}_{1c}(\alpha_1)$) with respect to 
rotation angles and power allocation parameters. 
Then, to obtain $\mathbf{Q}^{*}_2$ and ${R}^{*}_{2p}(\alpha_2)$ in \textit{Step 3a}, $\mathbf{Q}^{*}_1$ above is applied in Theorem~\ref{thm1}, and we solve the modified P2P MIMO problem using Lemma~\ref{waterfilling}. The precoding approach for \textit{Step 4} is the same as Scenario~A. The achievable secrecy rate for Scenario~B is given by the following corollary.
\begin{cor} \label{cor:corollary2} 
	The achievable rate region for secure MIMO-NOMA Scenario B under the total power is the convex hull of all rate triplets
	\begin{align}
	R_{{ B}}(P)= 
	\bigcup_{\alpha_k} ({R}^{*}_{1c}(\alpha_1),{R}^{*}_{2p}(\alpha_2),{R}^{*}_{0}(\alpha_0)).
	\end{align}
\end{cor}

\subsubsection*{Scenario C} (\textit{Step~1 $\rightarrow$ Step~2b $\rightarrow$  Step~3b $\rightarrow$ Step~4})

In  Scenario C, the steps are the same as  Scenario B except for \textit{Step 3b} which can be seen as a wiretap channel instead of P2P MIMO. Then, we apply Theorem~\ref{thm:theorem1} and solve \eqref{eq: derivation} instead. Similar to the precoding in \textit{Step 2b} of  Scenario~B, the covariance matrix $\mathbf{Q}_2$ can be written by rotation method as $\mathbf{Q}_2= \mathbf{V}_2 \mathbf{\Lambda}_2 \mathbf{V}_2^T$,
where the rotation matrix $\mathbf{V}_2$ is defined similarly to $\mathbf{V}_1$ in \eqref{eq_Vnbyn_} with rotation angles are
$\theta_{2pq}$. Therefore, the optimization problem for $R_{2c}(\alpha_2)$ becomes
\begin{subequations} \label{eq:r2star}
	\begin{align} 	
	R_{2c}(\alpha_2) 
	&=\max \limits_{\mathbf{Q}_2 \succeq \mathbf{0}} \frac{1}{2}\log 
	\frac{|\mathbf{I} + { \ddot{\mathbf{H}}_2 
			\mathbf{V}_2\mathbf{\Lambda}_2\mathbf{V}_2^{T} \ddot{\mathbf{H}}^{T}_2}|}{|\mathbf{I} +  \ddot{\mathbf{H}}_1
		\mathbf{V}_2\mathbf{\Lambda}_2\mathbf{V}_2^{T} \ddot{\mathbf{H}}^{T}_1|} 
	\label{eq:r2star2},\\
	&{\rm s.t.}\quad \sum_{n = 1}^{n_t} \lambda_{2n} \leq P_2 = (1-\alpha) 
	P,\label{eq:r2starconst}
	\end{align}
\end{subequations}
in which $\mathbf{\Lambda}_2 = {\rm diag}(\lambda_{21}, \dots, 
\lambda_{2n_t})$. This problem is again similar to \eqref{eq:r1star}.

In the power splitting scheme,  we solve $\mathbf{Q}_1^{*}$, $\mathbf{Q}_2^{*}$, and  $\mathbf{Q}_0^{*}$ to obtain $R_{1c}^{*}(\alpha_1)$, $R_{2c}^{*}(\alpha_2)$, and $R_0^{*}(\alpha_0)$  with respect to  power splitting parameters pair $(\alpha_1, \alpha_2, \alpha_0)$. 
Alternatively, we can first solve for $\mathbf{Q}_2^{*}$ followed by $\mathbf{Q}_1^{*}$ last $\mathbf{Q}_0^{*}$ (i.e., first $\bar{R}^{*}_{1c}(\alpha_1), 
\bar{R}^{*}_{2c}(\alpha_2)$, then $\bar{R}^{*}_{0}(\alpha_0)$). In general, changing the order of optimization will result in a different rate region. The \textit{convex hull} of the two solutions with different orders enlarges the achievable rate region. 
The achievable secrecy rate for 
 Scenario C is given by Corollary \ref{cor:corollary3}.

\begin{cor} \label{cor:corollary3} 
	The achievable DPC rate region for the secure MIMO-NOMA Scenario C under the total power is the convex hull of all rate triplets
	\begin{align}
	R_{{ C}}(P) = 
{\rm conv} \bigg\{ \bigg(\bigcup_{\alpha_k} \mathcal{R}_{12}(\alpha_k)  \bigg) \bigcup \bigg(  \bigcup_{\alpha_k} \mathcal{R}_{21}(\alpha_k)\bigg) \bigg\},
	\end{align}
\end{cor}
\noindent in which $\mathcal{R}_{12}(\alpha_k) = (R^{*}_{1c}(\alpha_1),  R^{*}_{2c}(\alpha_2), R^{*}_0(\alpha_0))$, $k=0,1,2$, is obtained by encoding the confidential messages for user 1 first, then user 2, and lastly the common message for both, whereas $\mathcal{R}_{21}(\alpha_k) = (\bar{R}^{*}_{1c}(\alpha), \bar{R}^{*}_{2c}(\alpha), \bar{R}^{*}_0(\alpha_0))$ is obtained in the reverse order of confidential messages (first user 2 then user 1).

\begin{algorithm}[t]
	\caption{Power splitting for all three scenarios}\label{alg:algorithmPS}
	\begin{algorithmic}[1]
		\STATE inputs: { secrecy scenario $L \in \{\textmd{A}, \textmd{B}, \textmd{C}\}$,} and $\epsilon_1$;
		\FOR{$\alpha_1 = 0:\epsilon_1:1$}
		\FOR{$\alpha_2 = 0:\epsilon_1:1-\alpha_1$}
		\STATE $\alpha_0 = 1-\alpha_1-\alpha_2$;
		\SWITCH {$L$}
		\CASE {{ $\textmd{A}$}}
		\STATE Obtain $\textbf{Q}_1^{*}$ using Lemma~\ref{waterfilling} in problem \eqref{eq:R1level1};
		\STATE Compute $R_{1p}$ in \eqref{eq:R1level1};
		\ENDCASE
		\CASE {$\textmd{B}$ or $\textmd{C}$}
		\STATE Obtain $\textbf{Q}_1^{*}$ by solving \eqref{eq:r1star} using BFGS;
		\STATE Compute $R_{1c}$ in \eqref{eq:R1};
		\ENDCASE
		\ENDSWITCH
		\SWITCH $L$
		\CASE {$\textmd{A}$ or $\textmd{B}$}
		\STATE Obtain $\textbf{Q}_2^{*}$ using Theorem~\ref{thm1}, the $\textbf{Q}_1^{*}$ in Line $7$ or Line $10$ with respect to $\textmd{A}$ or $\textmd{B}$, and Lemma~\ref{waterfilling} in problem \eqref{rate2p};
		\STATE Compute $R_{2p}$ in \eqref{rate2p};
		\ENDCASE
		\CASE {$\textmd{C}$}
		\STATE Obtain $\textbf{Q}_2^{*}$ using Theorem~\ref{thm:theorem1}, the $\textbf{Q}_1^{*}$ in Line $10$, and BFGS by solving \eqref{eq:r2star};
		\STATE Compute $R_{2c}$ in  \eqref{eq:R2};
		\ENDCASE
		\ENDSWITCH
		\STATE Compute $R_0$ as described in \textbf{\textit{Step 4}};
		\ENDFOR
		\ENDFOR
		\IF {{ $L=\textmd{A}$ or $L=\textmd{C}$}}
		\STATE swap all subscripts of 1 and 2 in  \eqref{eq: mathmode1} or \eqref{eq: mathmodel3};
		\STATE repeat \textbf{switch} and obtain $\mathcal{R}^{\rm DPC}_{21}(\alpha_k)$ or $\mathcal{R}_{21}(\alpha_k)$ in Corollary~\ref{cor:corollary1} and Corollary~\ref{cor:corollary3};
		\ENDIF
		\STATE outputs: $R_{\rm{L}}(P)$.
	\end{algorithmic}
\end{algorithm}

	Algorithm~\ref{alg:algorithmPS} summarizes the power splitting method for all of the scenarios. $\epsilon_1$ is the searching step for the power allocation factor.
	If $\alpha_1=\alpha_2=0$ and $\alpha_0\neq0$, then the system reduces to multicasting transmission. If no power is allocated to the common message, it is the private transmission cases in the next section \ref{sec:III}. If only the power of one of the secrecy messages is zero ($\alpha_k = 0$, $k=1$ or $2$), the problem is integrated service with confidential and common messages \cite{vaezi2019rotation}.

The precoding order for secrecy messages at different  scenarios is not the same.
Corollary~\ref{cor:corollary1} and Corollary~\ref{cor:corollary3} in Scenario~A and Scenario~C require an exchange of subscripts. For Scenario~A, this is because the encoding order affects the achievable rate region. For Scenario~C, although encoding order is irrelevant to the achievable rate in S-DPC, the order of optimization (solve the covariance matrix) will affect the solution \cite{qi2020secure}.
 This is because the power splitting method splits the power among the messages and solves them one by one. This simplifies the problem but is sub-optimal in general. Then, changing the precoding order may enlarge the achievable rate region.
For scenario B, as proved in \cite[Remark 4]{goldfeld2019mimo}, it is always better to cancel the private massage $M_{2p}$ at user 1 and treat the confidential message $M_{1c}$ at user 2 as noise \cite[Remark 4]{goldfeld2019mimo}. Thus, there is no need to exchange the precoding order.  


%
%
%



\begin{rem}[Complexity]
	For Scenario~A, \textit{Step~2a}, \textit{Step 3a}, Case~1, and Case~2 in \textit{Step~4} are analytical, which only requires the computation of matrix multiplications and matrix inverse. The computation of matrix multiplications and matrix inverse has the complexity of $\mathcal{O}(m^3)$ in which $m= \max(n_t, n_1, n_2)$.  Case~3 in \textit{Step 4} uses \texttt{fmincon} which is achieved mainly by BFGS. The BFGS algorithm 	yields the complexity $\mathcal{O}(n^2)$ \cite{nocedal2006numerical}, and the input variable $n=\frac{(n_t+1)n_t}{2}$ is rotation parameters \cite{zhang2020rotation}. Thus, the complexity of  Scenario~A is $\mathcal{O}(\frac{m^3+n_t^4}{\epsilon_1^2})$. For  Scenario~B and  Scenario~C, The complexity of solving wiretap channels in \textit{Step~2b} and \textit{Step~3b} has $\mathcal{O}(m^3+n_t^4)$. Ignore the coefficient, the overall complexity of Algorithm~\ref{alg:algorithmPS} is $\mathcal{O}(\frac{m^3+n_t^4}{\epsilon_1^2})$. The achievable DPC rate region in Scenario~A \cite{ekrem2010gaussian}, the capacity region achieved by DPC in Scenario~B \cite{goldfeld2019mimo}, and the capacity region achieved by S-DPC in Scenario~C \cite{ekrem2012capacity, Hung2013Liu} are found by using an exhaustive search over a set of positive semidefinite matrices, which have exponential complexity in terms of $m$. The three-dimensional space search in DPC or S-DPC has to be ``exhaustive'' but the search over power allocation factors is linear.

\end{rem}

\section{Weighted Sum Rate Formulation for Secrecy }\label{sec:III}

We consider the subcases of the three scenarios of the MIMO-NOMA without a common message ($M_0=\emptyset$) in this section. A weighted sum rate maximization based on BSMM \cite{park2015weighted, razaviyaynunified} is generalized to all  scenarios.
The WSR maximization for the MIMO-NOMA with private and confidential messages under a total power constraint is formulated as
\begin{align} \label{reformuP1}
\varphi(P)=&\max \limits_{\mathbf{Q}_k \succcurlyeq\mathbf{0}}
\sum_{k} w_k R_{k}, \quad k=1,2\notag \\
&{\rm s.t.} \quad {\rm tr}(\mathbf{Q}_k) \leq P,
\end{align}
where $R_{k} := R_{kp}$ in Scenario~A, $R_{1} := R_{1c}$ and $R_{2} := R_{2p}$ in  Scenario~B, and $R_{k} := R_{kc}$ in  Scenario~C. $w_k \geq 0$ is a weight. 
The Lagrangian of the problem  \eqref{reformuP1} is
\begin{align} \label{lagr}
L(\mathbf{Q}_1, \mathbf{Q}_2, \lambda)=w_1R_{1} + w_2R_{2}-\lambda({\rm tr}(\mathbf{Q}_1+\mathbf{Q}_2)-P),
\end{align}
where $\lambda$ is the Lagrange multiplier related to the total power constraint.
The dual function is a maximization of the Lagrangian 
\begin{align} \label{dual}
g(\lambda)  = \max \limits_{\mathbf{Q}_k \succcurlyeq\mathbf{0}} L(\mathbf{Q}_1, \mathbf{Q}_2, \lambda),
\end{align}
and the dual problem is given by
\begin{align} \label{dualp}
\min \limits_{\lambda \geq 0}g(\lambda){ .}
\end{align}
\begin{lem} \label{lemma1}
	The problem in \eqref{reformuP1} has zero duality gap and the KKT conditions are necessary  for the optimal solution.  
\end{lem}
\textit{Proof:} See the proof in Appendix A.

Since the problem in \eqref{dual} is a nonconvex problem in any  secrecy transmission, the BSMM \cite{razaviyaynunified, park2015weighted} can be considered which alternatively updates covariance matrix by maximizing a set of strictly convex local approximations. Specifically, Scenario~C has been studied in \cite{park2015weighted}. We discuss Scenario~A and Scenario~B in this paper.

\subsection{Scenario A}
It is worth noting that the MAC-BC  duality \cite{vishwanath2003duality} is applied to WSR maximization in \cite{viswanathan2003downlink} where the WSR on MAC rate region can be transformed to BC rate region by an iterative algorithm. Then, the WSR can be solved using  convex optimization. Once the optimum uplink covariance matrices are determined by any standard convex optimization tool, the equivalent downlink covariance matrices can be obtained through the duality transformation \cite{vishwanath2003duality}. The optimization in MAC requires a descent algorithm over a line search with a tolerance. It also mentions that the DPC rate region is difficult to compute without employing duality \cite{vishwanath2003duality}. Yet in this paper, we provide an alternative solution without applying the MAC-BC duality. We form a WSR for the DPC rate region directly and solve the maximization by using BSMM.

Since \eqref{dualp} is  nonconvex, we can apply  BSMM which updates covariance matrices by successively optimizing the lower bound of local approximation of
$f(\mathbf{Q}_1, \mathbf{Q}_2) = L(\mathbf{Q}_1, \mathbf{Q}_2, \lambda)$ \cite{razaviyaynunified, park2015weighted}. 
Rewrite $f(\mathbf{Q}_1, \mathbf{Q}_2)$ into the summation of one convex and one concave functions
\begin{align} \label{conv_conca}
f(\mathbf{Q}_1, \mathbf{Q}_2) =f_{1}(\mathbf{Q}_1) + f_{2}(\mathbf{Q}_1,\mathbf{Q}_2),
\end{align}
in which
\begin{subequations}\label{conv_conc1}
	\begin{align} 
	f_{1}(\mathbf{Q}_1)  = &{\frac{w_1}{2}} \log|\mathbf{I} +\mathbf{H}_1 \mathbf{Q}_1
	\mathbf{H}_1^T| - \lambda{\rm tr}(\mathbf{Q}_1) \\
	f_{2}(\mathbf{Q}_1,\mathbf{Q}_2)  =
	& {\frac{w_2}{2}}\log|\mathbf{I} + ({\mathbf{I} +\mathbf{H}_2 \mathbf{Q}_1 	\mathbf{H}_2^T})^{-1}\mathbf{H}_2 \mathbf{Q}_2 \mathbf{H}_2^T |   \notag \\
	&- \lambda({\rm tr}(\mathbf{Q}_2)-P).\label{sbcon_1}
	\end{align}
\end{subequations}
$f_{1}(\mathbf{Q}_1)$ is a concave function of  $\mathbf{Q}_1$, $f_{2}(\mathbf{Q}_1, \mathbf{Q}_2)$ is convex over $\mathbf{Q}_1$ by fixing  $\mathbf{Q}_2$.
After the decomposition, we can alternatively optimize $\mathbf{Q}_1$ and $\mathbf{Q}_2$ to find a lower bound for the weighted sum rate. For the $i$th iteration, the function for $f_{2}(\mathbf{Q}_1, \mathbf{Q}_1^{(i-1)})$ is lower-bounded by its first-order Taylor approximation \cite{boyd2004convex}
\begin{align} \label{Talyor}
f_{2}(\mathbf{Q}_1,\mathbf{Q}^{(i-1)}_2) \geq& f_{2}(\mathbf{Q}^{(i-1)}_1,\mathbf{Q}^{(i-1)}_2) \notag \\ - &{\rm {tr}}[\mathbf{A}	(\mathbf{Q}_1 - \mathbf{Q}^{(i-1)}_1)]
\end{align}
in which the power price matrix is a negative partial derivative with respect to $\mathbf{Q}_1$
\begin{align}
\mathbf{A} =& -\triangledown_{\mathbf{Q}_1} f_{2}(\mathbf{Q}_1^{(i-1)}, \mathbf{Q}^{(i-1)}_2)     \notag \\
=&-\frac{w_2}{\ln 2} \mathbf{H}_2^T (\mathbf{I} +\mathbf{H}_2( \mathbf{Q}_1^{(i-1)}+\mathbf{Q}^{(i-1)}_2)
\mathbf{H}_2^T)^{-1}\mathbf{H}_2\notag \\
&+\frac{w_2}{\ln 2} \mathbf{H}_2^T (\mathbf{I} +\mathbf{H}_2 (\mathbf{Q}^{(i-1)}_1)
\mathbf{H}_2^T)^{-1}\mathbf{H}_2.  \label{eq:A1}
\end{align}
Then the problem is lower bounded as
\begin{align} \label{reformuP21}
f(\mathbf{Q}_1, \mathbf{Q}^{(i-1)}_2) &\geq f_{1}(\mathbf{Q}_1) + f_{2}(\mathbf{Q}^{(i-1)}_1,\mathbf{Q}^{(i-1)}_2) \notag \\ & - {\rm {tr}}[\mathbf{A}	(\mathbf{Q}_1 - \mathbf{Q}^{(i-1)}_1)]{ .}
\end{align}
Then, we optimize the right-hand side of \eqref{reformuP21} by omitting the constant terms, which is equivalent as 
\begin{align} \label{reformuP31}
\mathbf{Q}^{(i)}_1= {\rm arg}\max \limits_{\mathbf{Q}_1} \; {\frac{w_1}{2}} \log|\mathbf{I} +\mathbf{H}_1 \mathbf{Q}_1
\mathbf{H}_1^T| - {\rm {tr}}[(\lambda \mathbf{I} + \mathbf{A})\mathbf{Q}_1]. 
\end{align}
Next, we optimize $f(\mathbf{Q}^{(i)}_1, \mathbf{Q}_2)$ by fixing $\mathbf{Q}^{(i)}_1$, which is equivalent as 
\begin{align} \label{reformuP41}
\mathbf{Q}^{(i)}_2= {\rm arg }\max \limits_{\mathbf{Q}_2} \; &{\frac{w_2}{2}}\log|\mathbf{I} + ({\mathbf{I} +\mathbf{H}_2 \mathbf{Q}^{(i)}_1 	\mathbf{H}_2^T})^{-1}\mathbf{H}_2 \mathbf{Q}_2 \mathbf{H}_2^T |   \notag \\
&- \lambda{\rm tr}(\mathbf{Q}_2).
\end{align}
The optimal solution for \eqref{reformuP31} and \eqref{reformuP41}  can be achieved by the following lemma \cite{park2015weighted}.

\begin{lem}{\cite{park2015weighted}} \label{lemopt}
	For some $\mathbf{S} \succ \mathbf{0}$, the optimal solution of the problem 
	\begin{align}
	\max_{\mathbf{Q}\succcurlyeq \mathbf{0}} w \log|\mathbf{I}+ \mathbf{R}^{-1}\mathbf{H}\mathbf{Q}\mathbf{H}^{T}|-{\rm tr}(\mathbf{SQ})
	\end{align} is given by 
	\begin{align}
	\mathbf{Q}^{*} = \mathbf{S}^{-1/2}\mathbf{V}\mathbf{\Lambda}\mathbf{V}^{T}\mathbf{S}^{-1/2}.
	\end{align}
		To use Lemma \ref{lemopt}, we set  $w= \frac{w_1}{2}$, $\mathbf{S}=\lambda \mathbf{I} + \mathbf{A}$, and $\mathbf{R}=\mathbf{I}$  for  \eqref{reformuP31};   and  $w= \frac{w_2}{2}$, $\mathbf{S}=\lambda \mathbf{I}$, and $\mathbf{R}={\mathbf{I} +\mathbf{H}_2 \mathbf{Q}^{(i)}_1 	\mathbf{H}_2^T}$ for \eqref{reformuP41}. $\mathbf{V}$, $\mathbf{U}$, and $\mathbf{\Lambda}$ are obtained by eigenvalue decomposition of  $\mathbf{R}^{-1/2}\mathbf{H}\mathbf{S}^{-1/2}=\mathbf{U}{\rm diag}(\sigma_1, \sigma_2,\cdots,\sigma_m)\mathbf{V}^{T}$, $\sigma_i \geq 0$, $\forall i$, $\mathbf{\Lambda	}= {\rm diag}[(w -1/\sigma_1^2)^{+}, \cdots, (w -1/\sigma_m^2)^{+}]$, and $(x)^{+}=\max(x,0)$.
	
\end{lem}

\subsection{Scenario B}
In Scenario B, what makes it different from Scenario~A is the formulation of convex and concave functions, which can be written as
\begin{subequations}\label{conv_conca1}
	\begin{align} 
	f_{1}(\mathbf{Q}_1)  = &{\frac{w_1}{2}} \log|\mathbf{I} +\mathbf{H}_1 \mathbf{Q}_1
	\mathbf{H}_1^T| - \lambda{\rm tr}(\mathbf{Q}_1) \\
	f_{2}(\mathbf{Q}_1,\mathbf{Q}_2)  =& -{\frac{w_1}{2}} \log|\mathbf{I} +\mathbf{H}_2 \mathbf{Q}_1 
	\mathbf{H}_2^T|\notag \\
	&+ {\frac{w_2}{2}}\log|\mathbf{I} + ({\mathbf{I} +\mathbf{H}_2 \mathbf{Q}_1 	\mathbf{H}_2^T})^{-1}\mathbf{H}_2 \mathbf{Q}_2 \mathbf{H}_2^T |   \notag \\
	&- \lambda({\rm tr}(\mathbf{Q}_2)-P).\label{sbconv_1}
	\end{align}
\end{subequations}
$f_{1}(\mathbf{Q}_1)$ is also a concave function of  $\mathbf{Q}_1$, $f_{2}(\mathbf{Q}_1, \mathbf{Q}_2)$ is convex by fixing  $\mathbf{Q}_2$ because the second term in \eqref{sbconv_1} is convex over $\mathbf{Q}_1$. For the $i$th iteration, the function for $f_{2}(\mathbf{Q}_1, \mathbf{Q}_1^{(i-1)})$ is lower-bounded by its first-order Taylor approximation as the expression in \eqref{Talyor},
in which the power price matrix
\begin{align} 
\mathbf{A} =& -\triangledown_{\mathbf{Q}_1} f_{2}(\mathbf{Q}_1^{(i-1)}, \mathbf{Q}^{(i-1)}_2)     \notag \\
=&\frac{w_1+ w_2}{2\ln 2} \mathbf{H}_2^T (\mathbf{I} +\mathbf{H}_2 \mathbf{Q}_1^{(i-1)}
\mathbf{H}_2^T)^{-1}\mathbf{H}_2\notag \\
&-\frac{w_2}{2\ln 2} \mathbf{H}_2^T (\mathbf{I} +\mathbf{H}_2 (\mathbf{Q}^{(i-1)}_1+\mathbf{Q}^{(i-1)}_2)
\mathbf{H}_2^T)^{-1}\mathbf{H}_2.  \label{eq:A}
\end{align}
Finally, we optimize the right-hand side of \eqref{reformuP21} with the power price matrix in \eqref{eq:A}, which is equivalent as 
\begin{align} \label{reformuP3}
\mathbf{Q}^{(i)}_1= {\rm arg}\max \limits_{\mathbf{Q}_1} \; {\frac{w_1}{2}} \log|\mathbf{I} +\mathbf{H}_1 \mathbf{Q}_1
\mathbf{H}_1^T| - {\rm {tr}}[(\lambda \mathbf{I} - \mathbf{A})\mathbf{Q}_1]. 
\end{align}
Next, we optimize $L(\mathbf{Q}^{(i)}_1, \mathbf{Q}_2)$ by fixing $\mathbf{Q}^{(i)}_1$, which is equivalent as 
\begin{align} \label{reformuP4}
\mathbf{Q}^{(i)}_2= {\rm arg }\max \limits_{\mathbf{Q}_2} \; &{\frac{w_2}{2}}\log|\mathbf{I} + ({\mathbf{I} +\mathbf{H}_2 \mathbf{Q}^{(i)}_1 	\mathbf{H}_2^T})^{-1}\mathbf{H}_2 \mathbf{Q}_2 \mathbf{H}_2^T |   \notag \\
&- \lambda{\rm tr}(\mathbf{Q}_2).
\end{align}

\begin{algorithm}[t]
	\caption{WSR maximization for all three scenarios without a common message}\label{alg:WSRalgorithm}
	\begin{algorithmic}[1]
		\STATE	inputs:  $\lambda^{\max}$, $\lambda^{\min}$,	$\epsilon_2$, 
		$\epsilon_3$, { secrecy scenario $L \in \{\textmd{A}, \textmd{B}, \textmd{C}\}$};
		\WHILE {$\lambda^{\max}-\lambda^{\min} > \epsilon_2$}
		\STATE$\lambda:=(\lambda^{\max}+\lambda^{\min})/2$;
		\STATE $\mathbf{Q}_1^{(0)}:=\mathbf{Q}_2^{(0)}:=\frac{P}{2n_t}\mathbf{I}$;
		\STATE $R^{(0)}:=0$;
		\STATE $i=0$;
		\WHILE 1
		\STATE $i=i+1$;
		\SWITCH {$L$}
		\CASE {$\textmd{A}$}
		\STATE Solve $\mathbf{Q}_1^{(i)}$ and $\mathbf{Q}_2^{(i)}$ in 
		\eqref{reformuP31}-\eqref{reformuP41} using 
		Lemma~\ref{lemopt};
		\STATE Compute $R_{1}$ and $R_{2}$ in \eqref{eq: mathmode1};
		\ENDCASE
		\CASE {$\textmd{B}$}
		\STATE Solve $\mathbf{Q}_1^{(i)}$ and $\mathbf{Q}_2^{(i)}$ in 
		\eqref{reformuP3}-\eqref{reformuP4} using 
		Lemma~\ref{lemopt};
		\STATE Compute $R_{1}$ and $R_{2}$ in \eqref{eq: mathmode2};
		\ENDCASE
		\CASE {$\textmd{C}$}
		\STATE Solve $\mathbf{Q}_1^{(i)}$ and 
		$\mathbf{Q}_2^{(i)}$ in 
		\cite[Algorithm 1, lines 5-13]{park2015weighted};
		\STATE Obtain $R_{1}$ and $R_{2}$ in \eqref{eq: mathmodel3};
		\ENDCASE
		\ENDSWITCH
		\STATE $R^{(i)}:=w_1R_{1}+ w_2R_{2}$ 
		\IF {${{\rm abs}(R^{(i)} - R^{(i-1)})} < \epsilon_3$}
		\STATE break;
		\ENDIF
		\IF {${\rm tr}(\mathbf{Q}_1^{(i)} + \mathbf{Q}_2^{(i)}) < P$}
		\STATE $\lambda^{\max}:=\lambda$;
		\ELSE
		\STATE $\lambda^{\min}:=\lambda$;
		\ENDIF
		\ENDWHILE
		\ENDWHILE
		\STATE outputs: $\lambda^{*}:=\lambda$, $R_{k}^{*}:=R_{k}$, and
		$\mathbf{Q}_k^{*}=\mathbf{Q}_k^{(i)}$, $k\in\{1,2\}$.
	\end{algorithmic}
\end{algorithm}

\begin{figure}[tb]
	\centering
	\includegraphics[width=0.38\textwidth]{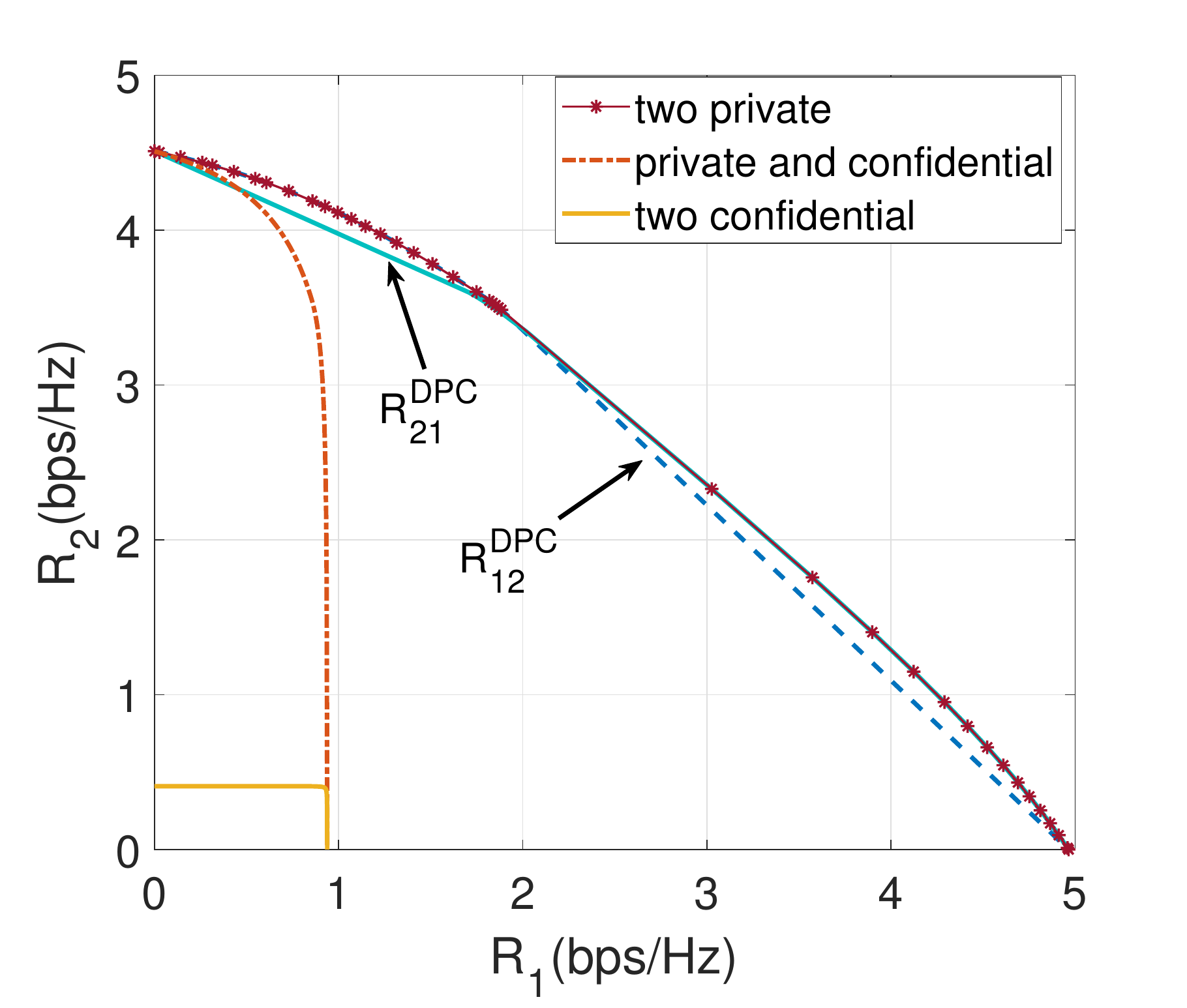}
	\caption{Secrecy capacity regions of three scenarios under an average total power
		constraint without common message over the channel $H_1=[0.3 \; 2.5;2.2 \; 
		1.8]$ and $H_2=[1.3 \; 1.2;1.5 \; 3.9]$, and $P=12$.}
	\label{fig:2}
\end{figure}
The WSR maximization for all scenarios without common message is summarized in Algorithm~\ref{alg:WSRalgorithm}. $\epsilon_2$ and $\epsilon_3$ are the bisection search accuracy and convergence tolerance of BSMM, respectively.	If $w_1=0$ and $w_2=1$, the problem reduces to a P2P MIMO with an analytical solution. Algorithm \ref{alg:WSRalgorithm} becomes a WF regime.
If $w_1=1$ and $w_2=0$, then problem in \eqref{reformuP1} is the secrecy rate maximization over MIMO wiretap channel. Then, Algorithm \ref{alg:WSRalgorithm} is nothing but AOWF \cite{li2013transmit}.
WSR maximization using BSMM can solve the special cases of the secrecy capacity regions in all three scenarios. However, it is hard to extend directly to the general cases of the three scenarios, because the max-min problem of multicasting is not derivable in BSMM although the multicasting problem owns convexity. 
	Thus, we propose a power splitting method for the general cases in Section~\ref{sec:IV}.

	Encoding order for secrecy messages in different scenarios is distinguished.
In Scenario A, the weight determines the optimal encoding order. For example, if $w_1 > w_2$, the optimal encoding order is encoding for user 1 of private message $M_{1p}$ first and then private message $M_{2p}$ for user 2 is encoded last. The capacity region is taken over all permutations of the users' order \cite{weingarten2004capacity}. In Scenario B, the entire capacity region is using DPC to cancel the signal of the private	message $M_{2p}$ intended for user 2 at user 1 only. 
The other variant which treats the private message $M_{2p}$ of user 2 as interference for user~1 is unnecessary \cite[Remark 4]{goldfeld2019mimo}.
 In Scenario C, the S-DPC owns the invariant property that the
achievable rate region is irrelevant with respect to the encoding order \cite{Hung2013Liu}. 

The three scenarios without common messages differentiate the security requirements. The achievable region is limited with a higher secrecy requirement. For comparison, we show an example in Fig.~\ref{fig:2} with the same channel settings as \cite{liu2010multiple, goldfeld2019mimo}. First, when the secrecy message of user 2 is zero, i.e., $M_{2p}=0$ in Scenario B and  $M_{2c}=0$ in Scenario~C, the maximal achieving rates for user 1 in the two cases are the same, and the two problems drop to the Gaussian wiretap channel. Second, when the secrecy message of user 1 is zero, i.e., $M_{1p}=0$ in Scenario~A and  $M_{1c}=0$ in  Scenario~B, the achieving rates for user 2 in the two cases become the same P2P MIMO problem.
Third, imposing a secrecy constraint on two users in  Scenario~C strictly shrinks the capacity region compared with  Scenario~A.

	\begin{rem}[Complexity]
		The number of iterations of the BSMM is $\mathcal{O}(1/\epsilon_3)$, and the bisection search requires  $\mathcal{O}(\log(1/\epsilon_2)$. The weighted sum rate of Algorithm~\ref{alg:WSRalgorithm}  has the complexity of $\mathcal{O}(\frac{m^3}{\sigma \epsilon_3 }\log({1}/{ \epsilon_2}))$ with a search step $\sigma$ over the weight \cite{park2015weighted, qi2020secure}.
		On the other hand, the computation complexity of Algorithm~\ref{alg:algorithmPS} without common messages is $\mathcal{O}(\frac{m^3+n_t^4}{\epsilon_1})$ with only one layer of search loop over $\alpha_1$. 	
	\end{rem}


\begin{figure*}[t] 
	\centering
	\subfigure[{ Scenario A}]{
		\includegraphics[scale=.243]{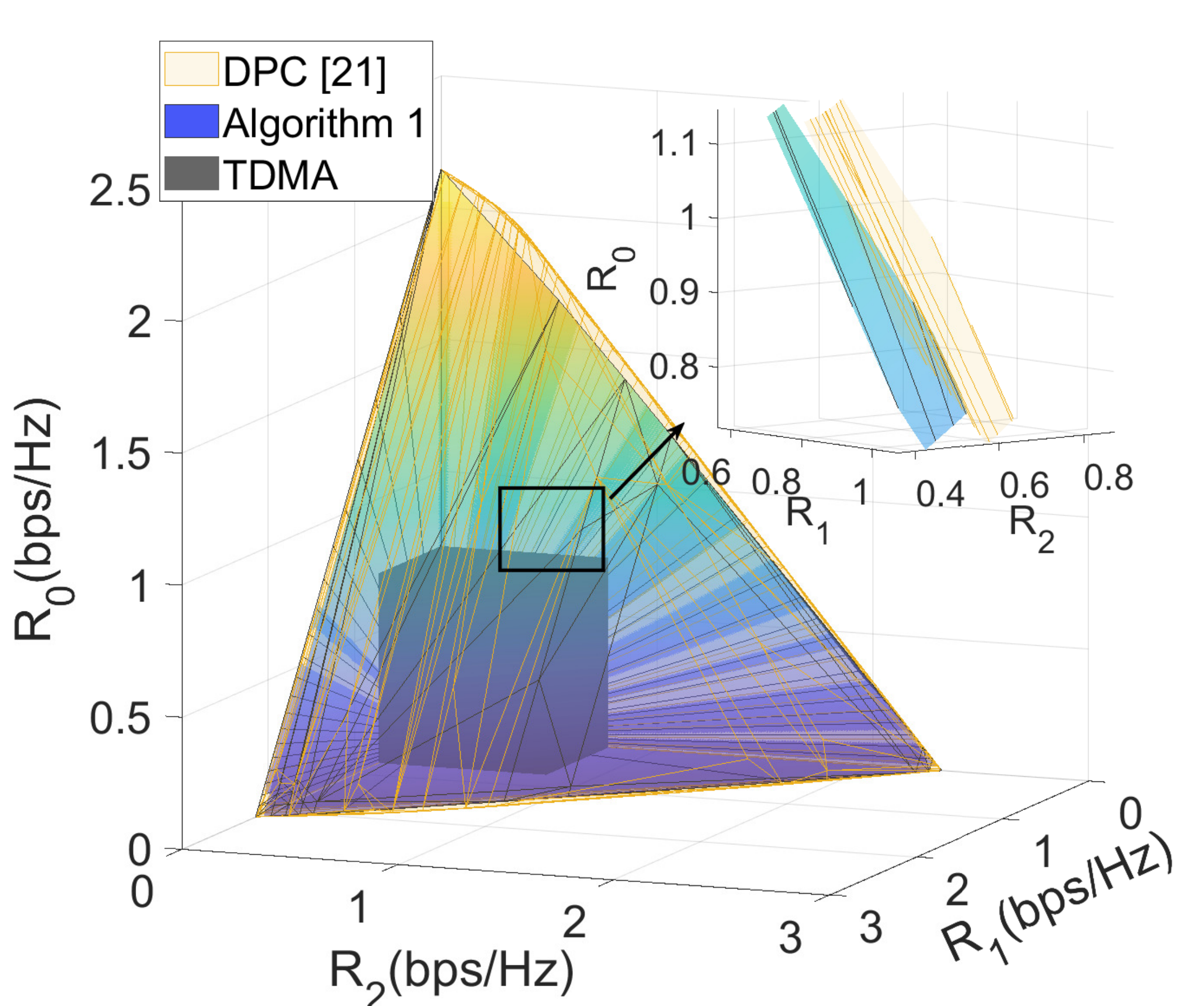}
	}
	\subfigure[{ Scenario B}]{
		\includegraphics[scale=.243]{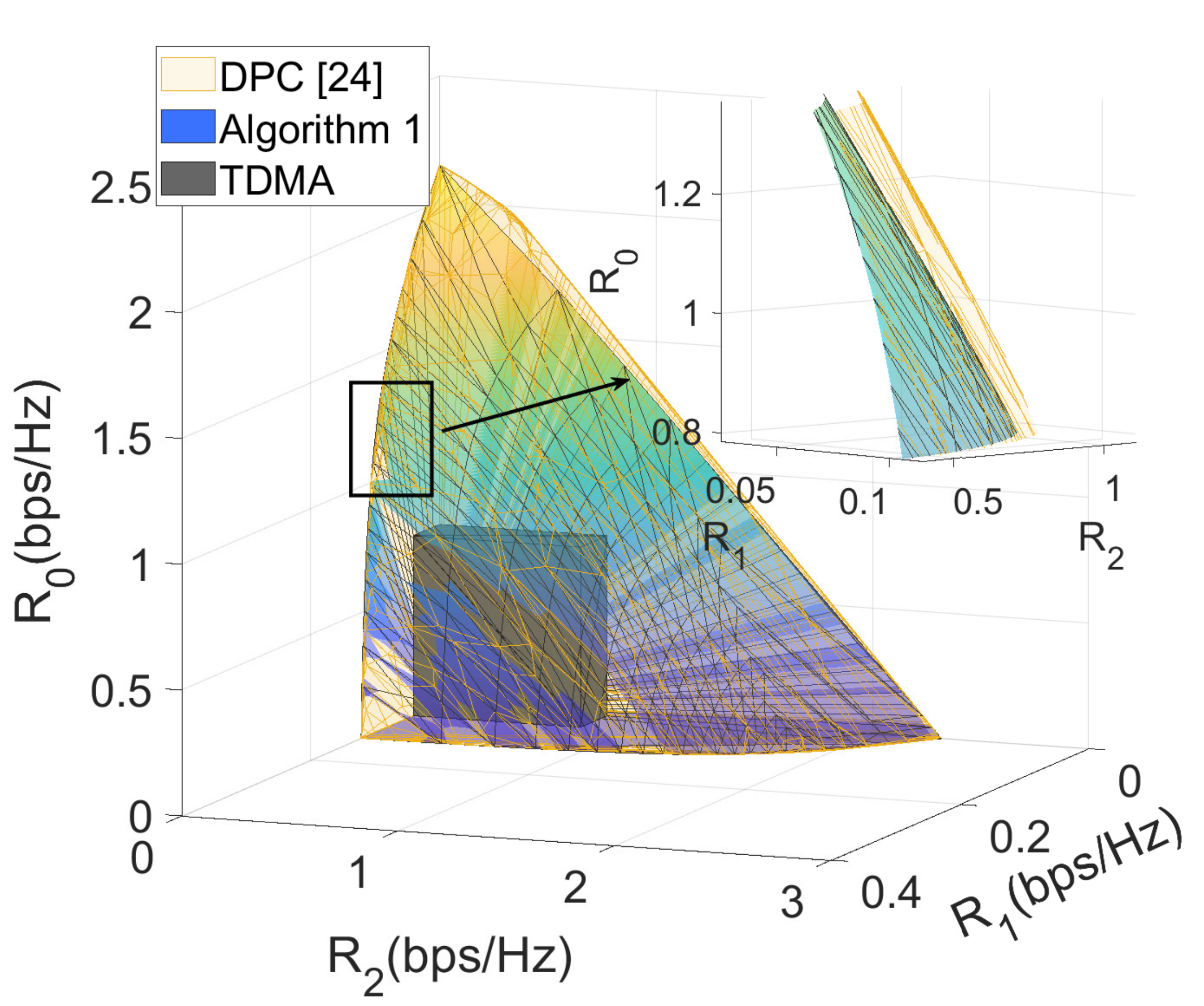}
	}
	\subfigure[{ Scenario C}]{
		\includegraphics[scale=.243]{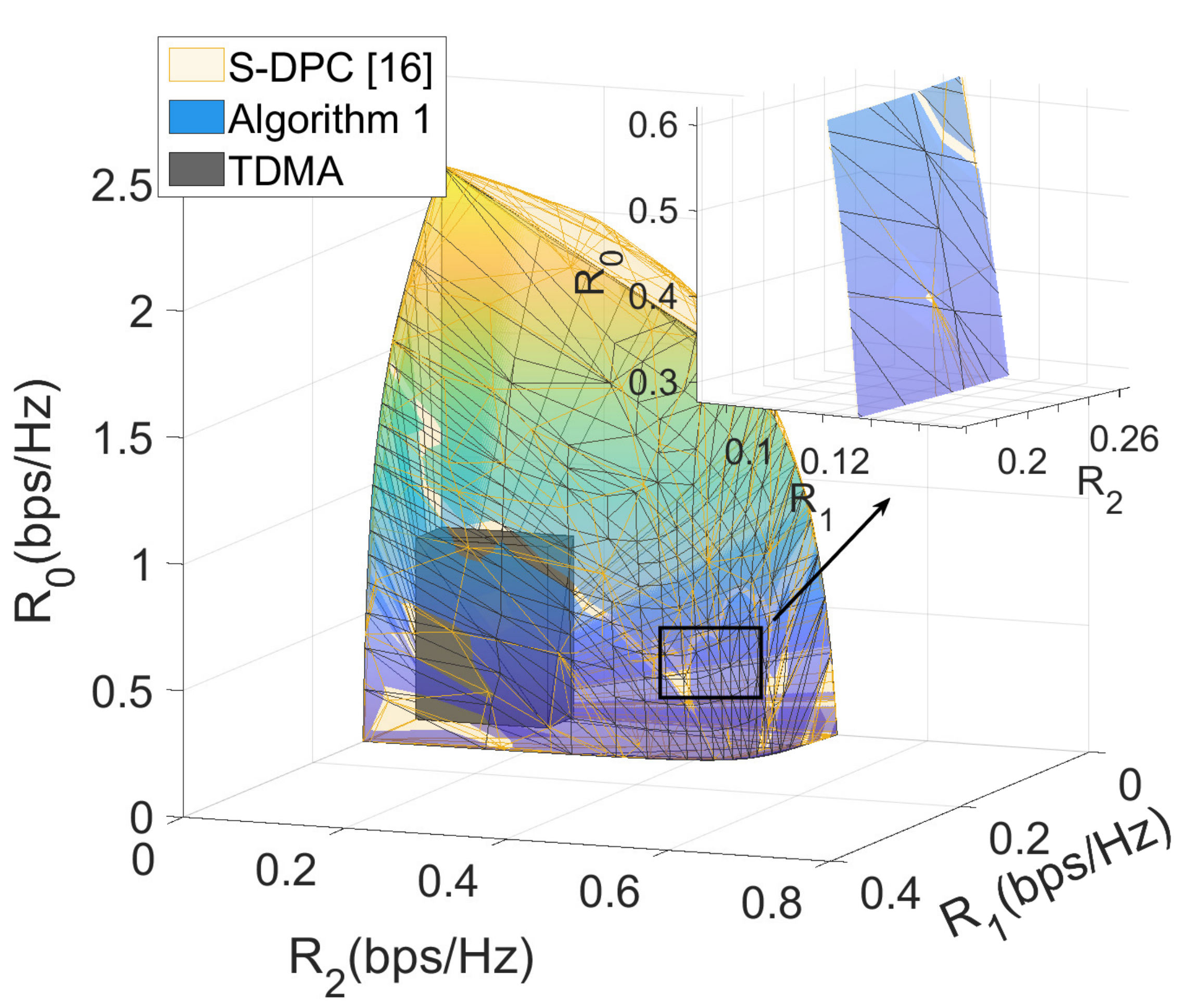}
	}
	\caption{Secrecy rate regions of MIMO-NOMA with different { scenarios of security} ($n_t=n_1=n_2=2$, and $P=10$). The yellow curved mesh is the secrecy capacity region, the colorful surface denotes the achievable rate region realized by Algorithm~\ref{alg:algorithmPS}, and the TDMA (gray cube) is achieved via three orthogonal time slots.}
	\label{fig:3levels}
\end{figure*}

\begin{figure*}[t] 
	\centering
	\subfigure[Scenario A]{
		\includegraphics[scale=.279]{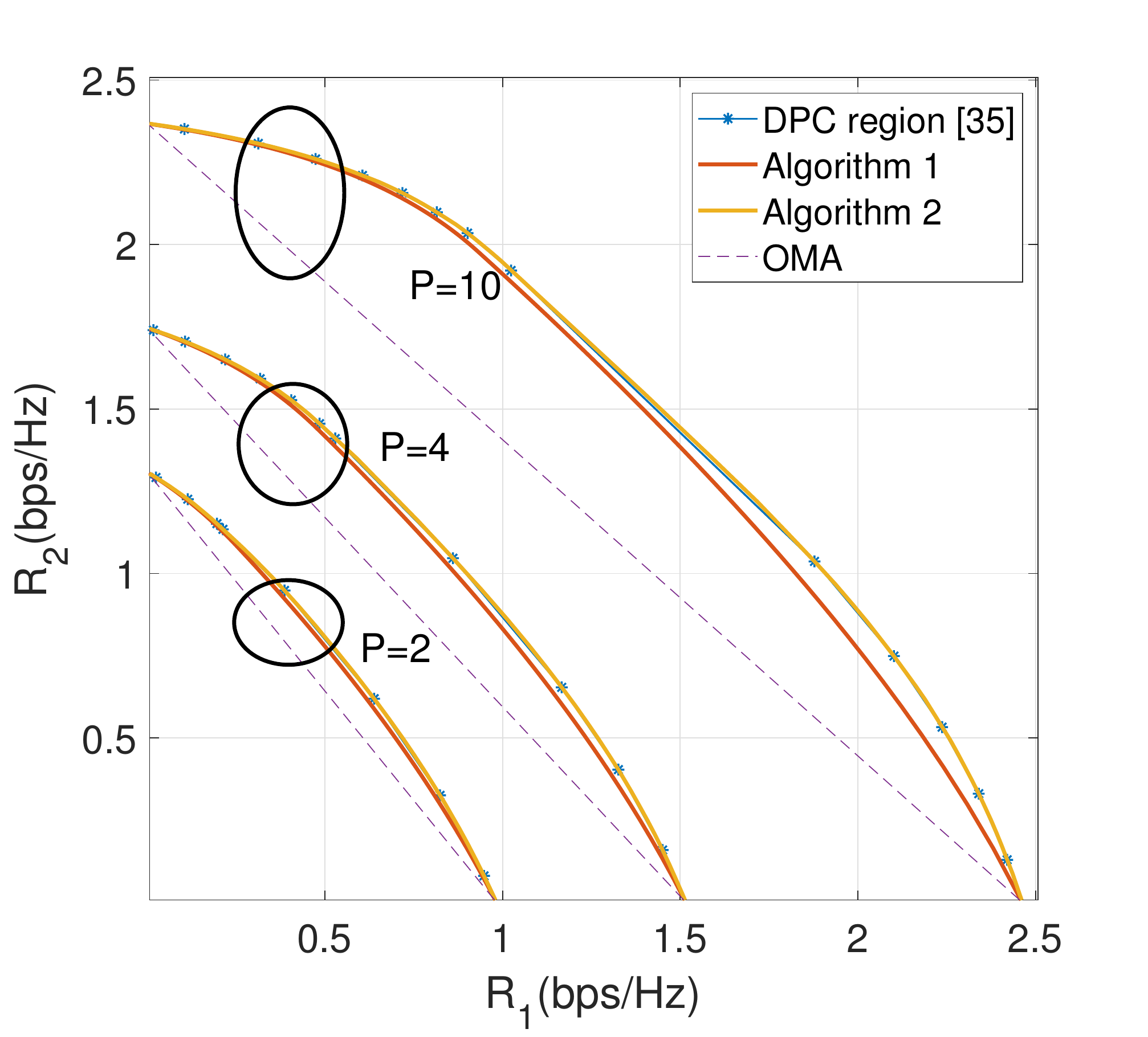}
				\label{fig:levelnoma1}
	}
	\subfigure[Scenario B]{
		\includegraphics[scale=.279]{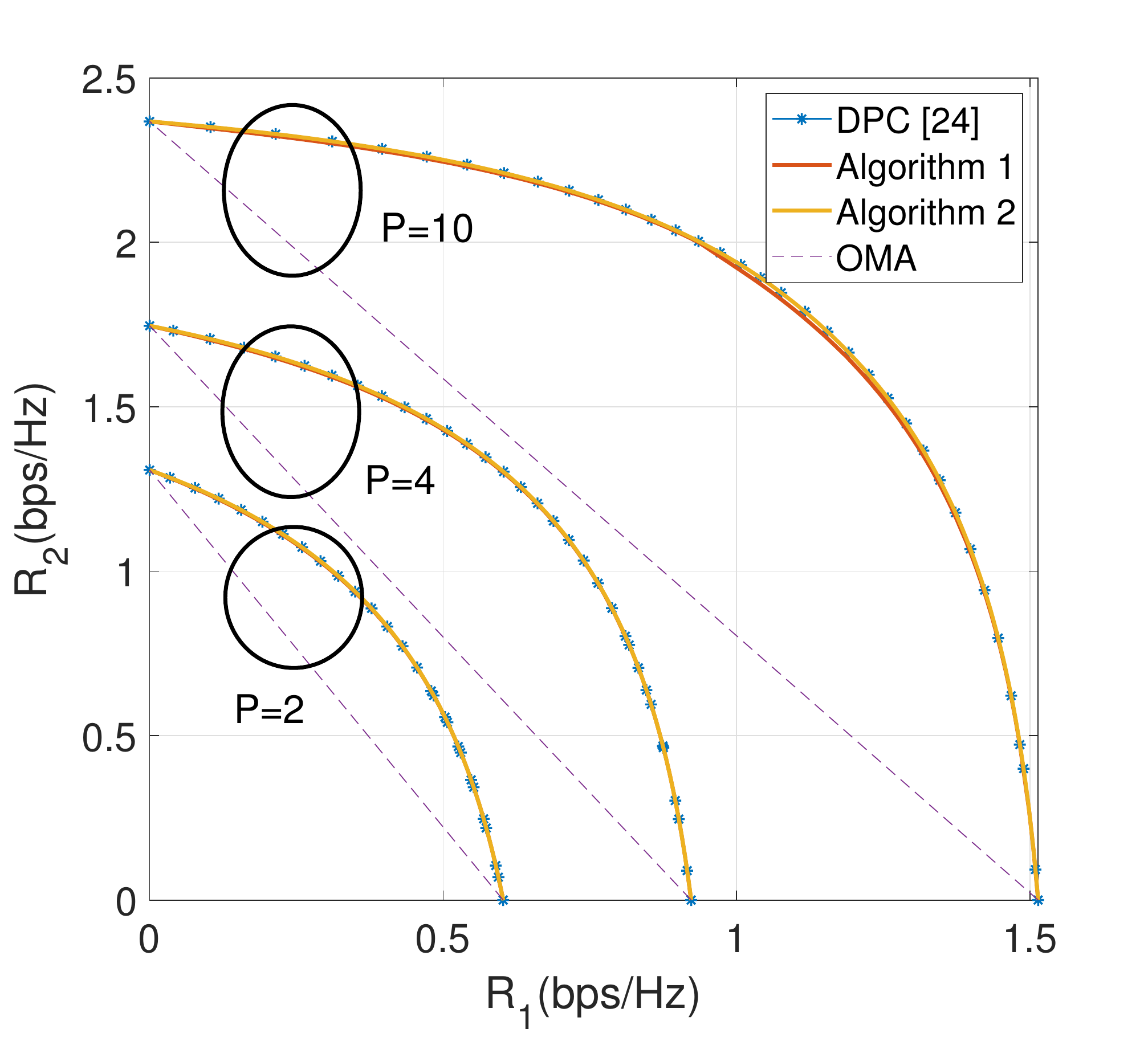}
				\label{fig:levelnoma2}
	}
	\subfigure[Scenario C]{
		\includegraphics[scale=.279]{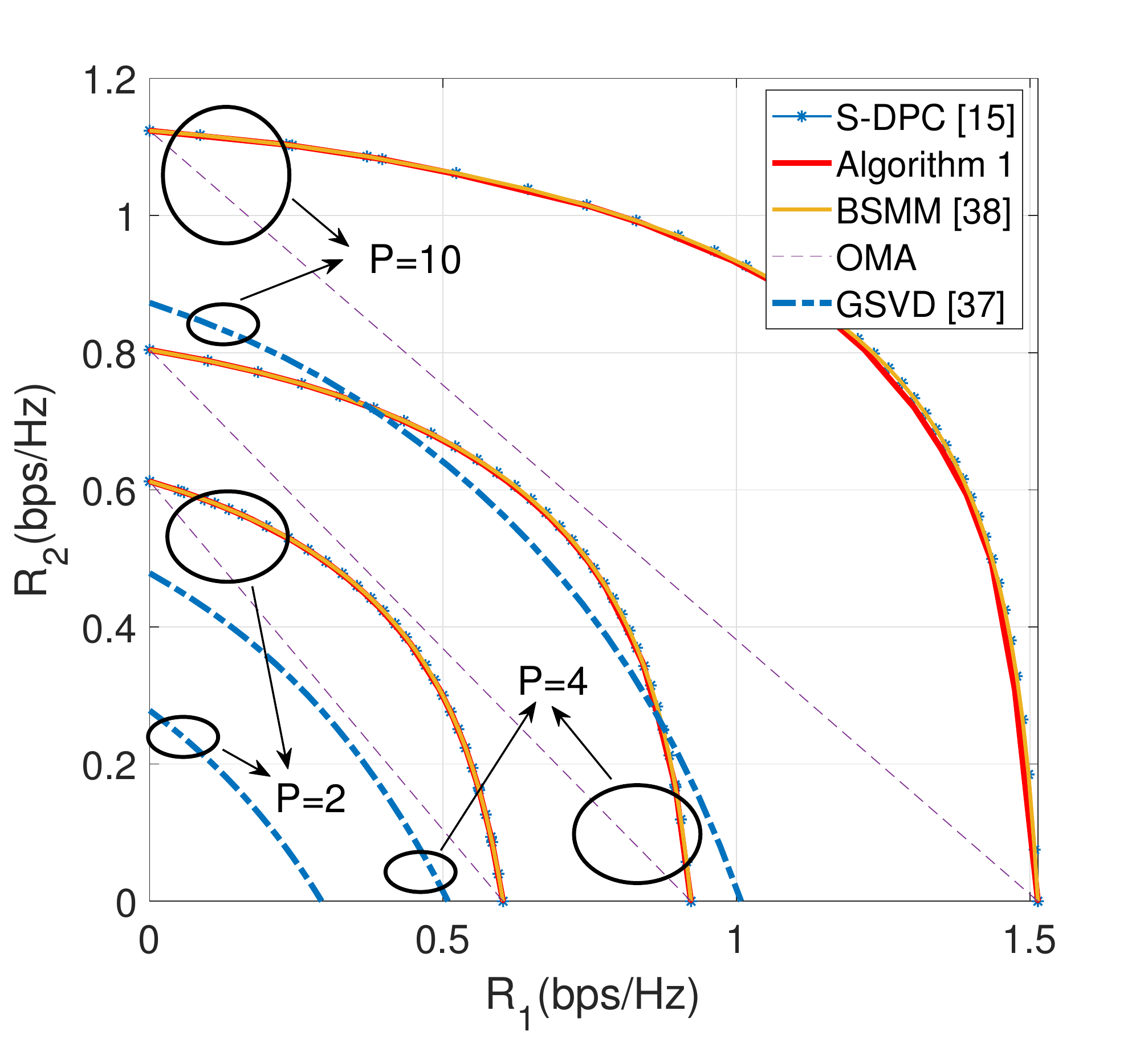}
				\label{fig:levelnoma3}
	}
	\caption{Secrecy rate regions of MIMO-NOMA  without multicasting services with different security requirements ($n_t=3$, $n_1=2$, $n_2=1$ and $P=2, 4, 10$). The blue dot line denotes the achievable or secrecy capacity region realized by DPC or S-DPC, the red line and yellow line are achieved by Algorithm~\ref{alg:algorithmPS} and  Algorithm~\ref{alg:WSRalgorithm}, respectively. The dash purple line is OMA reached by the
		time-sharing between the two extreme points \cite{qi2020secure}.}
	\label{fig:3levelsnocommon}
\end{figure*}

\section{Numerical results}\label{sec:V}
In this section, we perform numerical results to illustrate the achievable 
secrecy rate region of the three scenarios under the average power constraint and then verify Algorithm~\ref{alg:algorithmPS} and  Algorithm~\ref{alg:WSRalgorithm}.

\subsection{Secrecy Rate Regions for Three Scenarios}
First, we verify the transmission rates for all of the  scenarios.
In this simulation, the channels for 
	user~$1$ and user~$2$ are chosen to be
\begin{align}
\mathbf{H}_1 =\left[
\begin{matrix} 0.3861   & 0.6355\\
 0.9995  & 0.6259
\end{matrix}\right], \notag \;
\mathbf{H}_2 =\left[
\begin{matrix} 0.4977  &   0.9658\\
 0.9245   & 0.6116
\end{matrix}\right], \notag
\end{align}
where the channel coefficients	are generated independently according to Gaussian distribution,
and the total power is $10$. The search steps for 
$\alpha_1$ in Algorithm~\ref{alg:algorithmPS} is $0.05$. Fig.~\ref{fig:3levels} depicts the 
secrecy rate regions of the three scenarios. 
The PS scheme is compared with TDMA based scheme which is realized 
by transmitting messages in three orthogonal time slots with equal length. 
Also, the upper bounds are achieved by DPC \cite{ekrem2010gaussian,geng2014capacity} for Scenario~A, capacity rate regions \cite{goldfeld2019mimo}  and \cite{Hung2013Liu} for Scenario~B and C, respectively, which are realized by exhaustive search over all possible covariance matrices. It is 
shown 
that the proposed precoding and power allocation method significantly 
outperforms the 
TDMA strategy, and it is close to that of the capacity rate regions.  The projection of the secrecy capacity region  onto the
 $(R_1, R_2)$ or $(R_0, R_k)$, $k=1,2$, plane is the capacity region with two secrecy messages or only one 
secrecy message, which is going to appear in the next subsections.

It is worth mentioning that in Scenario A, given a set of power allocation parameters, we can analytically obtain the rate triplets, i.e., SVD and WF in \textit{Step~2a} and \textit{Step~3a}.  The complexity of the algorithm for finding one point on the region only comes from matrix operations, and  no search is needed. In \cite[Section III]{weingarten2006capacity} where each user is equipped with one antenna,  the rate maximization optimization is transferred to the power minimization problem, and thus a linear semi-definite convex optimization is obtained, but it needs a binomial search of one parameter and then apply one numerical method using standard semi-definite programming methods, e.g., \texttt{CVX} \cite{grant2009cvx}.

\subsection{Secrecy Rate Regions without Common Messages}
Consider the MIMO-NOMA case without common messages, the achievable rate region is realized by Algorithm~\ref{alg:algorithmPS} with $M_0 
= \emptyset$ and $\alpha_0=0$. Also, WSR with BSMM in Algorithm~\ref{alg:WSRalgorithm} is compared. The capacity regions   
 are achieved by the parameters including the search step  $0.01$, 
the total power $P=2, 4, 10$, respectively, and the channels denote as
\begin{align}
\mathbf{H}_1 &=\left[
\begin{matrix} 
 0.1560  & -0.6372  & -0.4055 \\
-1.1450  & -0.1417  &  0.0708
\end{matrix}\right], \notag\\
\mathbf{H}_2 &=\left[
\begin{matrix} -1.5032  &  0.5503 &  -0.0334
\end{matrix}\right]. \notag
\end{align}

Figure~\ref{fig:3levelsnocommon} compares the rate regions of the proposed 
power splitting scheme with the capacity region achieved by DPC for Scenario~A \cite{weingarten2004capacity,weingartens2006capacity}  generated using the iterative
	algorithm with MAC-BC duality presented \cite{viswanathan2003downlink},  and Scenario B \cite{goldfeld2019mimo}, respectively, and S-DPC \cite{liu2010multiple} for Scenario C. In general, the proposed algorithms can outperform OMA and achieve capacity.

In Scenario A, we consider another case when the numbers of receivers' antennas are limited to be the same, i.e., $n_1=n_2$. The channels are
	\begin{align}
	\mathbf{H}_1 =\left[
	\begin{matrix}  -1.3784  &  0.2593  & -0.2040\\
	-1.0689 &  -2.4811  & -1.2978
	\end{matrix}\right], \notag \\
	\mathbf{H}_2 =\left[
	\begin{matrix}  -0.3403  &  0.1358  & -1.9706\\
	-2.2982 &   -1.8135  &   0.2904
	\end{matrix}\right], \notag
	\end{align}
and $P=10$. From Fig.~\ref{fig:gsvd1}, the proposed algorithms can achieve a larger rate region than GSVD \cite{chen2019asymptotic} and OMA.
In Scenario C, the GSVD 
\cite{fakoorian2013optimality} and BSMM proposed by \cite{park2015weighted} are compared. The secrecy capacity region, i.e., DPC-based  rate
region, is obtained by \cite{vishwanath2003duality}. The proposed method can reach the secrecy rate region. Besides, the PS method in Algorithm~\ref{alg:algorithmPS} is  more  general for all  scenarios, while the WSR method in Algorithm~\ref{alg:WSRalgorithm} is specifically for the case without any broadcasting requirements.

\begin{figure}[htbp]
	\centering
	\begin{minipage}[t]{0.49\textwidth}
		\centering
		\includegraphics[scale=.46]{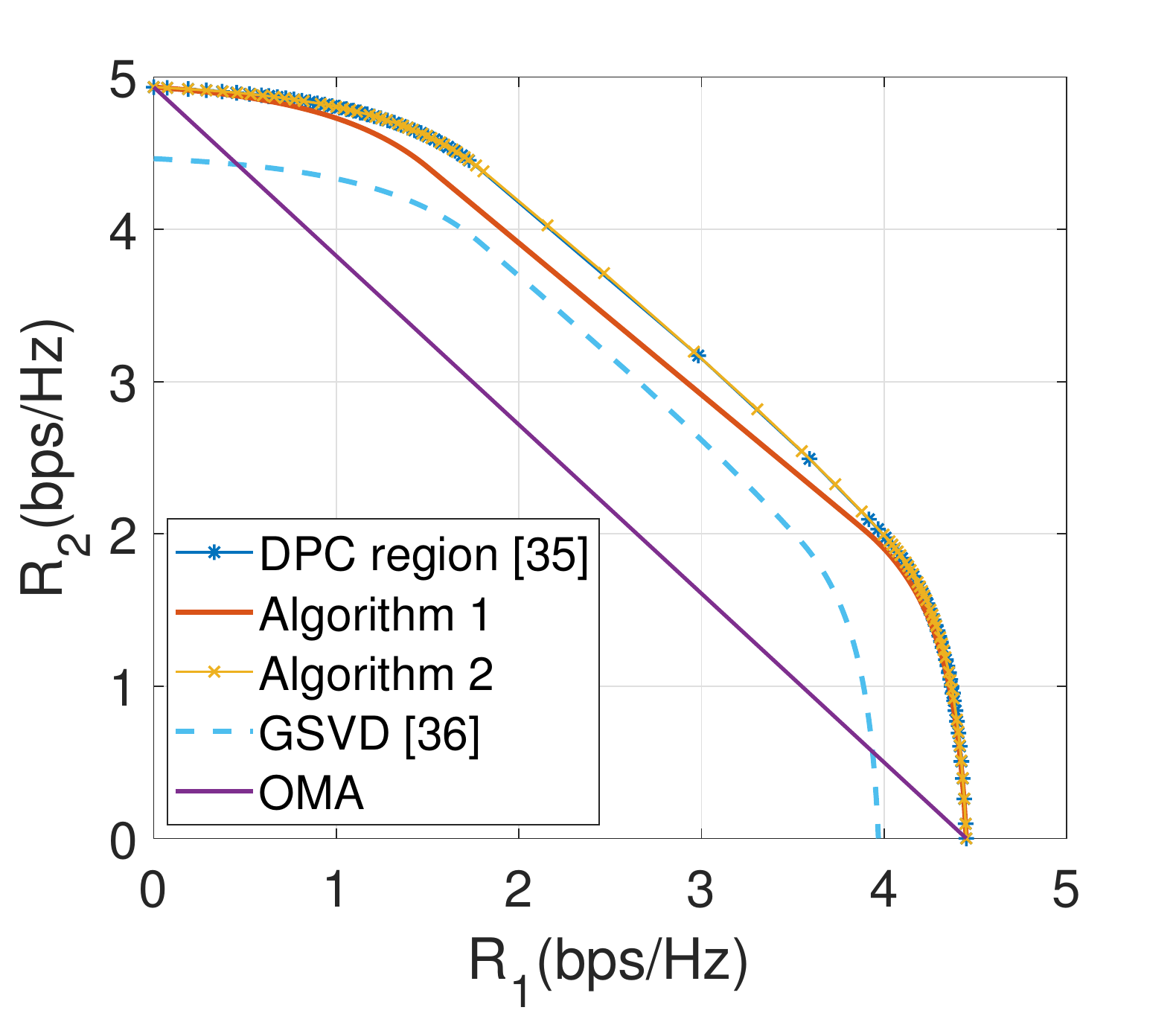}
		\caption{Comparison of the rate regions of Scenario A DPC \cite{weingarten2004capacity, weingartens2006capacity,viswanathan2003downlink}, GSVD \cite{chen2019asymptotic}, the proposed schemes, and OMA for $P=10$, and $n_t=3, 
	n_1=n_2=2$.}
	\label{fig:gsvd1}
	\end{minipage} 
	\begin{minipage}[t]{0.49\textwidth}
		\centering
		\includegraphics[scale=.46]{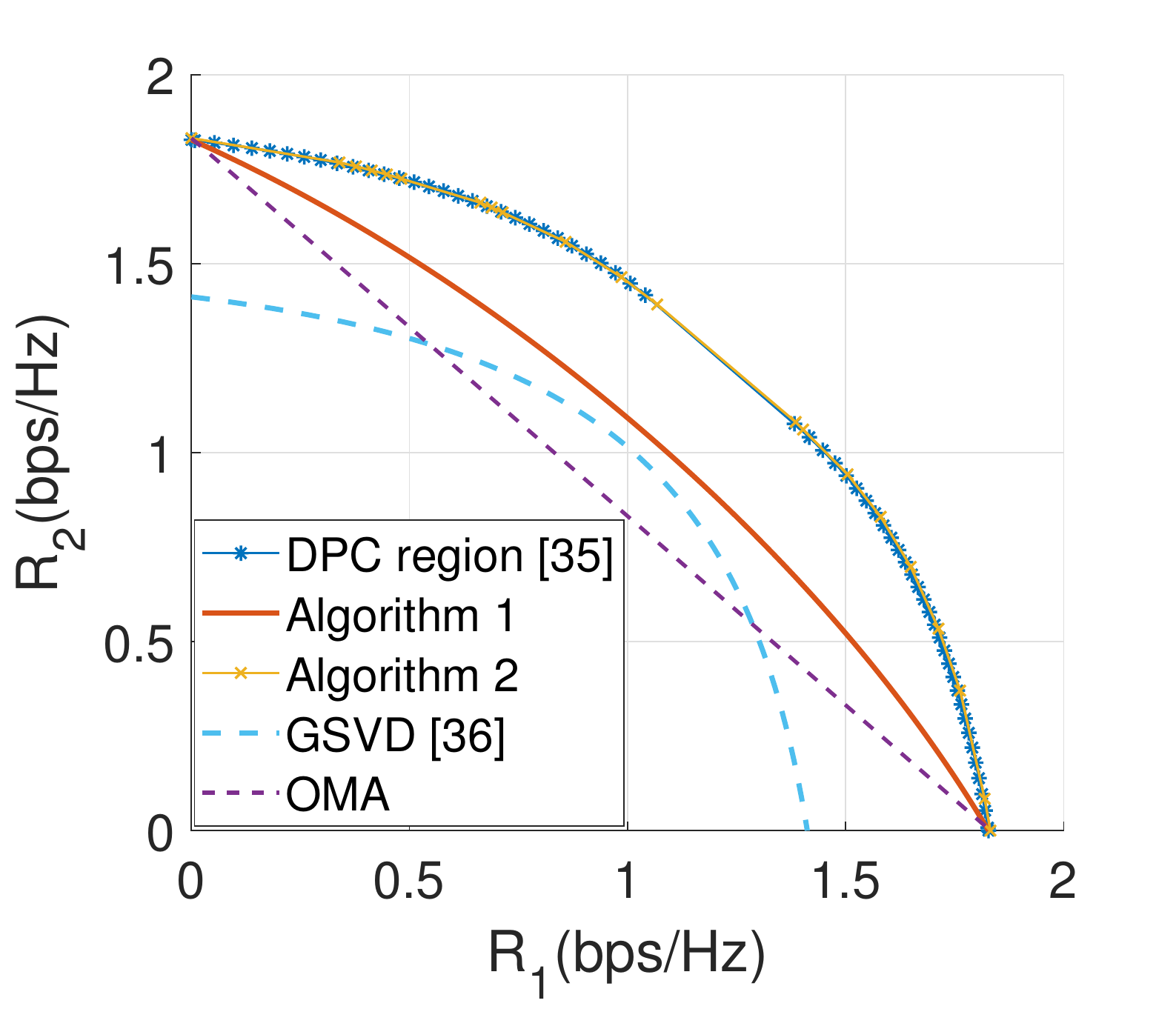}
	\caption{Comparison of the rate regions of Scenario~A DPC \cite{weingarten2004capacity,weingartens2006capacity,viswanathan2003downlink}, GSVD \cite{chen2019asymptotic}, the proposed schemes, and OMA for $P=10$, and $n_t=2, n_1=n_2=1$ in \cite{vishwanath2003duality}.}
	\label{fig:dpc1}
	\end{minipage}
\end{figure}


We provide one case with the same settings in \cite[Fig. 3]{vishwanath2003duality}, where the channels are:
	\begin{align}
		\mathbf{h}_1 =\left[
		\begin{matrix}  1  &  0.4 
		\end{matrix}\right],  \notag \;
		\mathbf{h}_2 =\left[
		\begin{matrix}  0.4 & 1
		\end{matrix}\right], \notag
	\end{align}
	and $P=10$. The results are shown in Fig.~\ref{fig:dpc1}. The iteration tolerance $t$ in \cite{viswanathan2003downlink} is set as $10^{-3}$, and a bisection search is applied to find the optimal $t$. We set our iteration accuracy  $\epsilon_2$ and convergence tolerance  $\epsilon_3$ in Algorithm 2 as $10^{-3}$. The complexity is the same because both methods require  finding the covariance matrices iteratively. The tolerance in \cite{viswanathan2003downlink} and the Lagrange multiplier in Algorithm 2 are both optimized through bisection search. Algorithm 1 is very fast without any search for one power allocation factor but is sub-optimal.

\subsection{Multicast and One Confidential Messages}
If we  set $\alpha_2= 0$ in Scenario C, then the general problem is reduced to the integrated services with one confidential and one common messages\footnote{One can also set $\alpha_1=0$ and change 
	the order of channels for $(R_0, R_{2c})$ which finally will resort to the same 
	results due to 
	duality.}, i.e., $(R_0, R_{1c})$. As  shown in Fig.~\ref{fig:r0r2}, the 
proposed method substantially outperforms the GSVD-based
orthogonal subchannel precoding method in \cite{weidong2016GSVD}, in which the turning point is a switch of subchannel selection schemes. Compared with the GSVD-based orthogonal subchannel decomposition, Algorithm~\ref{alg:algorithmPS} can make a better use of channel without decomposing the channel into many orthogonal subchannels.  Also, our method is very 
close to the secrecy capacity obtained by
rotation-based random exhaustive search
\cite{vaezi2019rotation}. In this simulation, search step for $\alpha_1$ is $0.05$, $P=15$, and 
channels are
\begin{align}
\mathbf{H}_1&=\left[
\begin{matrix}
0.0653   & 0.0185   & 1.0397\\
-0.1762   &-1.5297  &  0.1460\\
0.9822   &-1.9882   &-0.1263\\
0.9421  & -0.1771   & 0.3746
\end{matrix}\right], \notag\\
\mathbf{H}_2&=\left[
\begin{matrix}
-0.0248  &  1.3016 &   0.4677\\
0.0523  & -0.1297    &0.4269\\
0.6795  & -1.1725  & -0.8358
\end{matrix}\right]. \notag
\end{align}
\begin{figure}[t]
	\centering
	\includegraphics[width=0.4\textwidth]{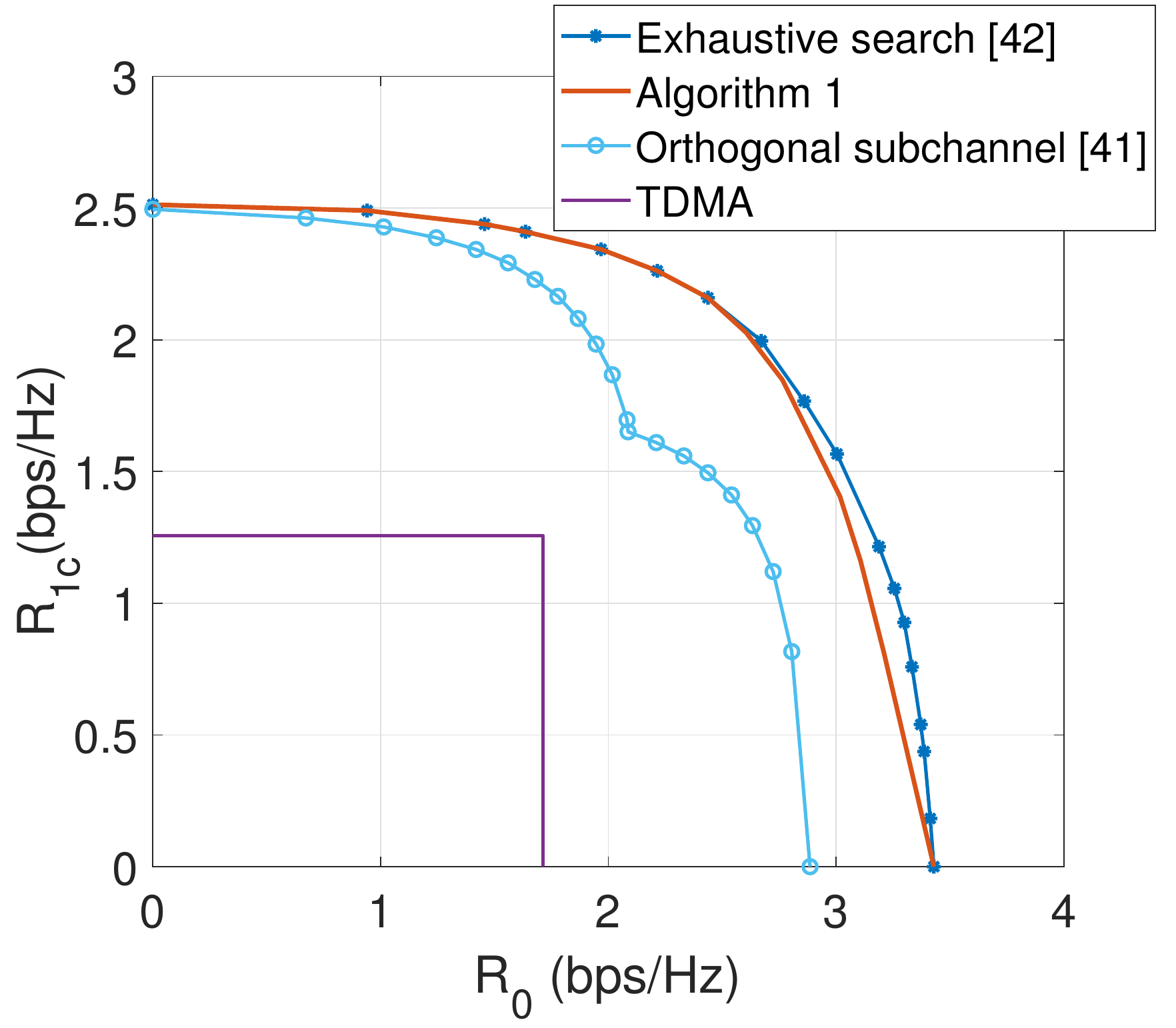}
	\caption{Comparison of the achievable rate regions of rotation-based exhaustive  search \cite{vaezi2019rotation}, orthogonal subchannel precoding \cite{weidong2016GSVD}, the proposed scheme, and TDMA for $P=15$, $n_t=3, n_1=4$, and $n_2=3$.}
	\label{fig:r0r2}
\end{figure}

\begin{table*}[!tb]
	\centering
	\caption{Comparison among different precoding schemes for the MIMO-NOMA with different  communication scenarios.}
	\label{summation}
	\begin{tabular}{|c|c|c|c|c|c|}
		\rowcolor[HTML]{EFEFEF}
		\hline
		& {\textbf{DPC/S-DPC}}              & \begin{tabular}[c]{@{}c@{}}\textbf{WSR with BSMM}\\ (proposed for Scenario~A, B\\without common message) \end{tabular}                   & \begin{tabular}[c]{@{}c@{}}\textbf{PS}\\ (proposed for all scenarios)\end{tabular}         & \textbf{GSVD}     & \textbf{OMA}       \\ \specialrule{.1em}{.05em}{.05em} 
		\textbf{Performance} & optimal & \begin{tabular}[c]{@{}c@{}}suboptimal\\ (but close to optimal)\end{tabular} & \begin{tabular}[c]{@{}c@{}}suboptimal\\ (but close to optimal)\end{tabular} & suboptimal & \begin{tabular}[c]{@{}c@{}}{highly} \\ suboptimal \end{tabular} \\ \hline
		\textbf{Speed}  &  generally slow            & \begin{tabular}[c]{@{}c@{}}acceptable for a small $m$ \\ {($m=\max(n_t, n_1, n_2)$)}\end{tabular}          &  fast for a small $n_t$      &  fast      &  very fast      \\ \hline
		\textbf{Complexity}  &        generally high           &    acceptable                &    acceptable        &      low       &      low      \\ \hline
		\textbf{Generality} &       	 \checkmark           &    \begin{tabular}[c]{@{}c@{}} not easy to generalize\\   for common message\end{tabular}           &   \checkmark          &      \checkmark        &      \checkmark        \\ \hline
	\end{tabular}
\end{table*}

 We notice that GSVD has been applied to many subcases. Examples are  two private messages in  Scenario A \cite{chen2019asymptotic}, two confidential messages in  Scenario C \cite{fakoorian2013optimality}, and one confidential message and one common message \cite{mei2016secrecy}. Thus, it also has the potential to become an efficient and general tool for all  scenarios. But, it should be noted that the performance of GSVD is affected by the number of antennas at the transmitter and users \cite{chen2019asymptotic, zhang2020rotation}. Algorithm~\ref{alg:WSRalgorithm} outperforms GSVD and sometimes Algorithm~\ref{alg:algorithmPS}, but it is not easy to extend it to common messages. Algorithm~\ref{alg:algorithmPS} balances the two methods.
We summarize the benefits and properties of the precoding schemes in Table~\ref{summation}.
Specifically, three linear precoding families in the MIMO-NOMA with different secrecy requirements are: 
\begin{itemize}
		\item GSVD: is the fastest general tool but has poor performance in some antenna settings. 
\item PS:  Algorithm~\ref{alg:algorithmPS} is a general suboptimal tool. It balances time and performance.
\item WSR:  Algorithm~\ref{alg:WSRalgorithm} is locally	 optimal with KKT as the optimal necessary conditions. It has relatively high time complexity. 
\end{itemize}

\section{Conclusions}\label{sec:conclusion}

In this paper, we investigate a two-user MIMO-NOMA network with different  security requirements. 
Specifically, three  scenarios are differentiated according to the required services:  multicast,  private, and/or confidential services. Although the capacity rate regions are known from information theory, the precoding and power allocation are still unclear for all three scenarios in general. A PS scheme is proposed which decomposes the MIMO BC into the P2P MIMO, the wiretap, and the multicasting channels. Then, existing solutions can be applied to obtain the precoding and power allocation matrices. The proposed PS can achieve near-capacity rate regions which are significantly higher compared to
the existing orthogonal methods. 
On the other hand, in the case of the MIMO-NOMA networks without multicasting, a WSR maximization based on BSMM is formulated for all three  scenarios. We generalize and prove that the duality holds for the WSR maximization, and the KKT conditions are necessary for the optimality. 
Numerical results demonstrate the performance of the WSR maximization. The two methods have their advantages. PS is a general tool for the MIMO-NOMA with different  scenarios of security, while WSR maximization provides a great potential for the secure MIMO-NOMA without multicasting, and both methods are computationally efficient compared with the DPC or S-DPC.

\appendices
\section{Proof of Lemma~\ref{lemma1}}
The duality gap is zero in Scenario A because the problem can be transferred as a convex problem satisfying Slater's condition  \cite{boyd2004convex}. 
Scenario C has zero duality gap and satisfies Lemma~\ref{lemma1} \cite[Theorem 1]{park2015weighted}. Now, we only need to prove that  Scenario B also has zero duality gap.

The \textit{time-sharing property} holds if the optimal value of the problem $\varphi(P)$ in \eqref{reformuP1} is concave of total power $P$, which implies a zero duality gap, i.e., the primal problem $\varphi(P)$
and the dual problem in \eqref{dualp} have the same optimal value \cite[Theorem 1]{yu2006dual}. The problem $\varphi(P)$  meets the time-sharing property because the secrecy capacity for Scenario~B without common message under total power constraint is a union of all the possible rates under covariance constraint \cite[Corollary 1]{goldfeld2019mimo}. The convex hull operation such as the time-sharing scheme will not enlarge the capacity region. 
Therefore,  the problem $\varphi(P)$ has zero duality gap and KKT conditions are necessary \cite{yu2006dual}.

\section*{Acknowledgment}
The authors are  grateful to the  reviewers for their suggestions to improve the quality of the paper. 

\bibliography{proposal1107}
\bibliographystyle{ieeetr}

\end{document}